\documentclass[11pt,oneside,openany]{article}

\usepackage[utf8]{inputenc}

\usepackage{geometry}
\usepackage{enumitem}
\setenumerate[1]{leftmargin=3em}
\usepackage{amssymb}
\usepackage{amsmath}
\usepackage{amsthm}
\usepackage{amsopn}
\usepackage[colorinlistoftodos]{todonotes}
\usepackage{mathrsfs}
\usepackage{mdframed}
\usepackage{tikz-cd}
\usepackage{subfig}
\usepackage{graphicx}
\usepackage{esint}
\usepackage{cancel}
\usepackage{thmtools, thm-restate}
\usepackage{hyperref}

\usepackage{verbatim}
\usepackage{mathtools}
\usepackage{fullpage}
\usepackage{dsfont}

\usepackage{physics}
\usepackage{setspace}
\usepackage{epstopdf}
\usepackage[T1]{fontenc}
\usepackage{graphicx}

\usepackage{scalerel,stackengine}
\usepackage{ulem}

\usepackage{authblk}

\usepackage{soul}

\usepackage[numbers]{natbib}

\usepackage{marginnote}

\newcommand{\mytodo}[2][]{{%
		\let\marginpar\marginnote
		\reversemarginpar
		\renewcommand{\baselinestretch}{0.8}%
		\todo[#1]{#2}}}

\theoremstyle{remark}
\newtheorem{Remark}{Remark}[section]

\theoremstyle{plain}
\newtheorem{Theorem}{Theorem}[section]
\newtheorem{Lemma}{Lemma}[section]

\newtheorem{Definition}{Definition}[section]
\newtheorem{Proposition}{Proposition}[section]

\DeclareRobustCommand{\rchi}{{\mathpalette\irchi\relax}}
\newcommand{\irchi}[2]{\raisebox{\depth}{$#1\chi$}} 

\DeclareMathOperator{\supp}{supp}

\def \qed {\hfill \vrule height6pt width 6pt depth 0pt}

\newcommand{\R}{\mathbb{R}}
\newcommand{\N}{\mathbb{N}}

\newcommand{\C}{\mathbb{C}}

\newcommand{\p}{\partial}

\newcommand{\doublehookrightarrow}{%
\mathrel{\mathrlap{{\mspace{4mu}\lhook}}{\hookrightarrow}}
}

\usepackage{ulem}

\usepackage{natbib}



\allowdisplaybreaks

\numberwithin{equation}{section}

\stackMath
\newcommand\reallywidehat[1]{%
\savestack{\tmpbox}{\stretchto{%
\scaleto{%
\scalerel*[\widthof{\ensuremath{#1}}]{\kern-.6pt\bigwedge\kern-.6pt}%
{\rule[-\textheight/2]{1ex}{\textheight}}
}{\textheight}%
}{0.5ex}}%
\stackon[1pt]{#1}{\tmpbox}%
}

\newcommand\numberthis{\addtocounter{equation}{1}\tag{\theequation}}

\renewcommand{\bar}{\overline}
\renewcommand{\hat}{\widehat}
\renewcommand{\tilde}{\widetilde}
\renewcommand{\leq}{\leqslant}
\renewcommand{\geq}{\geqslant}

\title{Convergence towards the Vlasov-Poisson Equation \\from the $N$-Fermionic Schr\"odinger Equation}
\date{}
\author[1]{Li Chen\thanks{chen@math.uni-mannheim.de}}
\author[2,3]{Jinyeop Lee\thanks{lee@math.lmu.de}}
\author[1]{Matthew Liew\thanks{mliew@mail.uni-mannheim.de}}
\affil[1]{Institut f\"ur Mathematik, Universit\"at Mannheim}
\affil[2]{School of Mathematics, Korea Institute for Advanced Study}
\affil[3]{Mathematisches Institut, Ludwig-Maximilians-Universit\"at M\"unchen}

\begin{document}

	\maketitle

\begin{abstract}
We consider the quantum dynamics of $N$ interacting fermions in the large $N$ limit.
The particles in the system interact with each other via repulsive interaction that is regularized Coulomb potential with a polynomial cutoff with respect to $N$.
From the quantum system, we derive the Vlasov-Poisson system by simultaneously estimating the semiclassical and mean-field residues in terms of the Husimi measure. 
\end{abstract}

\normalem

\section{Introduction}

In this study, we consider a system of $N$ identical spinless fermions characterized by the wave function $\psi_{N}:\R^{3N} \to \C$ in $L^2_a (\R^{3N})$ with $\norm{\psi_{N}}_{L^2}^2 =1$. The antisymmetric space $L^2_a (\R^{3N})$, which is a subspace of $L^2 (\R^{3N})$, is given by

\begin{equation}\label{def_asymm_space}
		L^2_a(\R^{3N}) := \left\{ \psi_N \in L^2(\R^{3N}) : \psi_N (x_{\pi(1)}, \dots, x_{\pi(N)}) = \varepsilon_\pi  \psi_N (x_1, \dots, x_N),\ \text{for all }\pi \in S_N  \right\},
	\end{equation}
where $S_N$ is the odd-permutation group and $\varepsilon_\pi$ is the sign of the permutation $\pi$.

The antisymmetric space considered above is a reflection of fermions obeying the Pauli exclusion principle, i.e. no two identical fermions simultaneously occupy the same single quantum state. It is observed that when $N$ fermions are initially trapped in a volume of order one, their kinetic energy is at least of order $N^{5/3}$. This implies that the coupling constant should be chosen as $N^{-1/3}$ to balance the order of the potential energy and the kinetic energy in the Hamiltonian. Thus, the mean-field Hamiltonian acting on $L^2_a(\R^{3N})$ is given by
	\begin{equation*}
		H_N = - \frac{1}{2} \sum_{j=1}^N  \Delta_{x_j} + \frac{1}{2N^{1/3}} \sum_{i \neq j}^N V_N(x_i - x_j),
	\end{equation*}
where $ \Delta_{x_j}$ is the Laplacian acting on particle $x_j$ and $V_N$ is the interaction potential given by the regularized Coulomb potential defined as follows:

\begin{Definition}\label{assume_V_N} For any $x \in \R^3$ and let $V(x) = |x|^{-1}$, then we call the following $V_N$ to be the regularized Coulomb potential:
	\begin{equation}\label{def_V_N}
		V_N (x)= (V*\mathcal{G}_{\beta_N})(x),
	\end{equation}
	where
	$\mathcal{G}_{\beta_N}(x) := \frac{1}{ (2\pi \beta_N^2)^{3/2}} e^{-\left({x/\beta_N}\right)^2}$. 
\end{Definition}

The regularized Coulomb potential defined in \eqref{def_V_N} can be understood as an interaction potential between spherical particles with a vanishing radius $\beta_{N}\to 0$ as $N \to \infty$. This method of using the regularized Coulomb potential depending on $N \to \infty$ has been applied in many works, for example, in \cite{jabin2011particles,Lazarovici2017} for the derivation of the Vlasov-Poisson dynamics from $N$-body classical dynamics. In \cite{Chen2011}, such a regularized potential was considered for the bosonic case.

Observe that, the time-dependent Schr\"odinger equation is given by
	\[
	\mathrm{i} \p_\tau \psi_{N,\tau}  = H_N \psi_{N,\tau},
	\]
for all $\psi_{N,\tau} \in L^2_a(\R^{3N})$ and $\tau \geq 0$. Since the average kinetic energy for each fermionic particle is of order $N^{2/3}$, then its average velocity is of order $N^{1/3}$. Therefore, in the mean-field regime, the time evolution of the fermion system is expected to be of order $N^{-1/3}$. Rescaling the time variable $t = N^{1/3} \tau$, one obtains the  following Schr\"odinger equation for $N$ fermions:
\begin{equation}\label{schro_0}
		N^{1/{3}}\mathrm{i} \partial_t \psi_{N,t} =H_N \psi_{N,t}.
	\end{equation}

As suggested in Thomas-Fermi theory in \cite{lieb1997thomas,lieb1973thomas}, we set $\hbar = N^{-1/3}$ as the semiclassical scale. Then, multiplying both sides of \eqref{schro_0} by $\hbar^2$, we obtain the time-dependent Schr\"odinger equation as follows:\footnote{Note that $\hbar$ here can be interpreted as the effective Planck's constant.}
\begin{equation}\label{Schro_1}
		\begin{cases}
			\displaystyle \mathrm{i} \hbar \partial_t \psi_{N,t} = \left[ - \frac{\hbar^2}{2} \sum_{j=1}^N  \Delta_{x_j} + \frac{1}{2N} \sum_{i \neq j}^N V_N(x_i - x_j) \right]  \psi_{N,t}, \\
			\psi_{N,0} = \psi_N,
		\end{cases}
	\end{equation}
where $\psi_N$ is the initial data in $L_a^2(\R^{3N})$. The choice of other coupling constants for different scenarios is summarized in \cite{BACH20161}. 

Solving numerically the Schr\"odinger equation in \eqref{Schro_1} with a large particle number and analyzing the behavior of its solution is hard even for $N=1000$. An efficient way to analyze and solve the behavior of a large quantum system is to derive its corresponding effective evolution equations. Therefore, we consider the density matrix operator instead of the wave function $\psi_{N,t}$. Namely, for $t\geq 0$, we define the $1$-particle reduced density matrix $\gamma_{N,t}$, a positive semidefinite trace class operator in $L_a^2(\R^{3N})$, with trace equal to $N$. Specifically, for pure states, it is an operator with the corresponding kernel given by
\begin{equation*}
	\gamma_{N,t}^{(1)}(x;y) := N\dotsint \mathrm{d}x_{2} \cdots \mathrm{d}x_N\  \overline{\psi_{N,t} (y, x_{2},\dots x_N) } \psi_{N,t} ( x, x_2, \dots, x_N),
	\end{equation*}
for any normalized $\psi_{N,t} \in L_a^2(\R^{3N})$.
It can be easily shown that the trace of the $1$-density particle is given by $\Tr \gamma_{N,t}^{(1)} = N$.
Furthermore, for indistinguishable fermions, we can analyze the quantum dynamics by density matrices depending on a small number of particles, $1 \leq k \ll N$. Denoting $\Tr^{(k)}$ as the $k$-partial trace, we define the $k$-particle reduced density matrix as
\begin{equation}\label{tr_gammak}
		\gamma_{N,t}^{(k)}  = \frac{N!}{(N-k)!} \Tr^{(k)} \gamma_{N,t},
	\end{equation}
where its corresponding integral kernel is given by
\begin{align*}
	&\gamma_{N,t}^{(k)} (x_1,\dots, x_k ; y_1, \dots, y_N) \\
	&\hspace{1cm}= \frac{N!}{(N-k)!} \dotsint \dd{x}_{k+1} \cdots \dd{x}_N \gamma_N(x_1, \dots,x_k, x_{k+1}, \dots, x_N; y_1, \dots, y_k, x_{k+1}, \dots, x_N).
	\end{align*}

We denote the inner-product of $L_a^2(\R^{3N})$ as $\left< \psi, \phi \right> = \int \mathrm{d}x \, \bar{\psi}(x) \phi(x)$. Given any $N$ and time $t$, the expectation of the physical observable associated with a self-adjoint operator $O$ is given as
\[
	\left< \psi_{N,t}, O \psi_{N,t} \right> = \dotsint  \mathrm{d}x_{1} \cdots \mathrm{d}x_N\ \ \overline{\psi_{N,t}(x_1, \dots, x_N)} \big(O \psi_{N,t}\big)(x_1, \dots, x_N).
	\]
Equivalently, we can write the expectation of an observable $O$ with
\begin{equation}\label{exp_value_obs}
		\Tr O \gamma_{N,t} = \left<\psi_{N,t} , O \psi_{N,t}  \right>,
	\end{equation}
and the expectation of any $k$-observables $O^{(k)}$ is
\[
	\Tr (O^{(k)} \otimes \mathds{1}^{(N-k)})\gamma_{N,t} =  \frac{N!}{(N-k)!} \Tr O^{(k)} \gamma_{N,t}^{(k)}.
	\]
Therefore, the $k$-particle reduced density matrix $\gamma_N^{(k)}$ is also a positive semidefinite trace class operator with trace
\[
	\Tr \gamma_{N,t}^{(k)} = \frac{N!}{(N-k)!}.
	\]

With a $k$-particle density matrix, we can avoid analyzing the complicated case with $N$-particles by finding an approximating effective equation that describes the system. In the fermionic case, we let $\gamma_{N,0}^{(1)} \equiv \omega_{N}$, a 1-particle density matrix associated with initial state $\psi_{N}$, be a Slater determinant defined as 

\begin{equation}\label{def_slater}
	\begin{split}
		\psi_{N}^{\text{Slater}}(x_1, \dots, x_N) &= (N!)^{-{1/2}} \det\{\mathrm{e}_i(x_j)\}_{i,j = 1}^N,\\
	\end{split}
\end{equation}
for any family of orthonormal bases $\{\mathrm{e}_j\}_{j=1}^N \subset L^2(\R^3)$. In particular, we have
\begin{equation}\label{1_slater}
		\omega_{N} =  \sum_{j=1}^N \dyad{\mathrm{e}_j}{\mathrm{e}_j} , 
	\end{equation}
which corresponds to the $1$-particle reduced density matrix with an integral kernel of $\omega_{N} (x;y)= \sum_{j=1}^ N  \overline{\mathrm{e}_j(y)} \mathrm{e}_j(x)$.
In \cite{Porta2017}, the mean-field approximation of the Schr\"odinger equation is given by the following Hartree-Fock equation:
\begin{equation}\label{main_HF_eq}
	\begin{cases}
		\displaystyle \mathrm{i} \hbar \partial_t \omega_{N,t} = \left[ - \hbar^2 \Delta + ( |\cdot|^{-1} *\rho_{N,t}) - X_t,\ \omega_{N,t} \right], \\
		\omega_{N,t} \big|_{t=0} = \omega_{N},
	\end{cases}
	\end{equation}
where $\rho_{N,t}$ has the integral kernel $\frac{1}{N} \omega_{N,t} (x;x)$, $X_t$ is the exchange operator with the integral kernel $\frac{1}{N} |x-y|^{-1} \omega_{N,t}(x;y)$ and the commutator is denoted as $[A,B] := AB - BA$ for any bounded operators $A$ and $B$.

The mean-field limit from the Schr\"odinger equation to the Hartree-Fock equation has been studied extensively. In \cite{ELGART20041241}, where the Slater determinant constitutes the initial data and a regular interaction is assumed, the convergence is obtained by the use of the Bogoliubov-Born-Green-Kirkwood-Yvon (BBGKY) hierarchy method for short times. In \cite{benedikter2014mean}, the rates of convergence in both the trace norm and Hilbert-Schmidt norm for pure states are obtained for an arbitrary time and more general potential in the framework of second quantization. The extension to mixed states has been considered in \cite{benediktermixed} for a positive temperature and for the relativistic case in \cite{benedikter2014rel}. Furthermore, by utilizing the Fefferman-de la Llave decomposition presented in \cite{BACH20161,Fefferman1986,hainzl2002general}, the rate of convergence, with more assumptions on the initial data is obtained in \cite{Porta2017} for  Coulomb potential and in \cite{saffirio2017mean} for inverse power law potential. Further literature on the mean-field limit for fermionic cases can be found in \cite{Frohlich2011,petrat2014derivation,Petrat2017,Petrat2016ANM}.

The semiclassical limit from the Hartree-Fock equation to the Vlasov equation has also been extensively studied.
In \cite{Lions1993}, this is achieved by using the Wigner-Weyl transformation of the density matrix. In \cite{benedikter2016hartree}, the authors compared the inverse Wigner transform of the Vlasov solution and the solution of the Hartree-Fock equation and obtained the rate of convergence in the trace norm as well as the Hilbert-Schmidt norm with regular assumptions on the initial data. In fact, \cite{benedikter2016hartree, saffirio2019hartree} utilized the $k$-particle Wigner measure as follows:
\begin{equation}\label{def_wigner}
	\begin{aligned}
		&W^{(k)}_{N,t}(x_1,p_1, \dots, x_k, p_k) \\
		&= \binom{N}{k}^{-1} \dotsint (\dd{y})^{\otimes k}  \gamma_{N,t}^{(k)}\left(x_1 + \frac{\hbar}{2}y_1, \dots, x_k + \frac{\hbar}{2}y_k; x_1 - \frac{\hbar}{2}y_1, \dots, x_k - \frac{\hbar}{2}y_k \right)e^{-\mathrm{i} \sum_{i=1}^k p_i \cdot y_i},
	\end{aligned}
\end{equation}
where $\gamma_{N,t}^{(k)}$ is the kernel of the $k$-particle reduced density defined in \eqref{tr_gammak}.

The works in this direction have also been extended for the inverse power law potential in \cite{Saffirio2019}, rate of convergence in the Schatten norm in \cite{lafleche2020strong}, Coulomb potential and mixed states in \cite{saffirio2019hartree}, and convergence in the Wasserstein distance in \cite{Lafleche2019GlobalSL,Lafleche2019PropagationOM}.
The convergence of relativistic Hartree dynamic to the relativistic Vlasov equation was considered in \cite{Dietler2018}. Further analysis of the semiclassical limit from the Hartree-Fock equation to the Vlasov equation can be found in \cite{Athanassoulis2011, amour2013classical,amour2013,Gasser1998, Markowich1993}.

We can combine both mean-field and semiclassical limits and directly obtain the convergence from the Schr\"odinger equation to the Vlasov equation. The notable pioneers in this direction are Narnhofer and Sewell in \cite{Narnhofer1981} and Spohn in \cite{Spohn1981}. They proved the limit from the Schr\"odinger equation to Vlasov, in which the interaction potential $V$ was assumed to be analytic in \cite{Narnhofer1981} and $C^2$ in \cite{Spohn1981}. The rate of convergence of the combined limit in terms of the Wasserstein (pseudo)distance was obtained in \cite{Golse2017,golse:hal-01334365,Golse2021}. In fact, the authors studied the rate of convergence in terms of the Wasserstein distance by treating the Vlasov equation as a transport equation and applying the Dobrushin estimate with appropriately chosen initial data. Then, the result for the Husimi measure was obtained by transforming its Wigner measure similar to \eqref{gaussian_m_vs_W} with a specifically chosen coherent state. In this study, we instead consider a more generalized coherent state.
Recently, the combined limit for the singular potential case was obtained in \cite{chong2021schrodinger}. They provided a derivation of the Vlasov equation using the weighted Schatten norm with a higher moment, and more conditions on the initial data were assumed.

Nevertheless, it is known that the Wigner measure defined in \eqref{def_wigner} is not a true probability density, as it may be negative in a certain phase space. This is shown numerically in \cite{kenfack2004negativity} for chosen Fock states. Moreover, in \cite{PhysRevLettKatz}, a \textit{vis-à-vis} comparison of the classical and quantum systems of a nonlinear Duffing resonator shows that the classical system develops a probability density in the traditional sense, while the quantum system yields a negative region in phase space corresponding to the Wigner measure. In fact, it is proven in \cite{Hudson1974,PhysRevA.79.062302,doi:10.1063/1.525607} that the Wigner measure is nonnegative if and only if the pure quantum states are Gaussian. Additionally, in \cite{doi:10.1063/1.531326}, it is stated that the Wigner measure is nonnegative if the state is a convex combination of coherent states. The issue of incompatibility between the quantum Wigner and classical regimes remains an open question \cite{CaseWignerpedestrians}. 

Nevertheless, it has been shown that we can obtain a nonnegative probability measure by taking the convolution of the Wigner measure with a Gaussian function as a mollifier; this is known as the Husimi measure \cite{Fournais2018,Combescure2012,Zhang2008}. In particular, from \cite[p.21]{Fournais2018}, given a specific Gaussian coherent state, the relation between the Husimi measure and Wigner measure is given by the following convolution: for any $1 \leq k \leq N$,
\begin{equation}\label{gaussian_m_vs_W}
		\mathfrak{m}_{N,t}^{(k)}  = \frac{N(N-1)\cdots (N-k+1)}{N^k} W^{(k)}_{N,t} * \mathcal{G}^\hbar,
	\end{equation}
where $\mathfrak{m}_{N,t}^{(k)}$  is the $k$-particle Husimi measure and
\begin{equation*}
	\mathcal{G}^\hbar := (\pi \hbar)^{-3k} \exp \big(-\hbar^{-1} (\sum_{j=1}^k |q_j|^2 + |p_j|^2) \big).
\end{equation*}

The smoothing of the Wigner measure presented in \eqref{gaussian_m_vs_W} motivates the objective of our study: to directly obtain the Vlasov-Poisson equation from the Schr\"odinger equation in terms of the Husimi measure.\footnote{See Figure \ref{nice_figure}.}  In fact, we have explored the direct method in \cite{Chen2021} with the use of the BBGKY hierarchy method, under the assumption that $V \in W^{2,\infty} (\R^3)$. The main contribution of the current work is that by using the generalized version of Husimi measure defined later in \eqref{husimi_def_1}, we are able to write $N$-fermionic Schr\"odinger equation directly into Vlasov type of equation in \eqref{BBGKY_k1} and obtain a convergence in combined-limit without the use of BBGKY hierarchy method. Furthermore, compared to \cite{Chen2021}, the new remainder terms in \eqref{BBGKY_k1} obtained in this paper allows us to handle the regularized Coulomb potential defined in \eqref{def_V_N}.

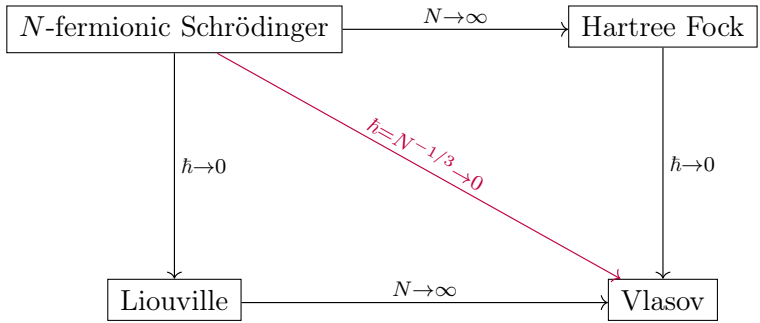
\begin{figure}[t]
\begin{center}
\begin{tikzcd}[column sep=3cm,row sep=3cm,cells={nodes={draw=black}}]
\text{$N$-fermionic Schr\"odinger} \arrow[r, "N \to \infty"] \arrow[d,"\hbar \to 0"] \arrow[rd,"\hbar=N^{-1/3}\to 0"{sloped,auto,pos=0.5},purple]
& \text{Hartree Fock} \arrow[d, "\hbar \to 0"] \\
\text{Liouville}\arrow[r, "N \to \infty"]& \text{Vlasov}
\end{tikzcd}
\end{center}
\caption{Relations of $N$-fermionic Schr\"odinger systems to other mean-field equations \cite{Chen2021,Golse2016,Golse2017}.}\label{nice_figure}
\end{figure}

Note that the case for bosons has been extensively studied. In fact, there are more studies on bosonic cases than on fermionic cases. As bosons are not the main concern of this paper, we mention only a selected few of these studies in passing. In particular, \cite{erdos2001derivation} proved that the Schr\"odinger equation for bosons converges to the nonlinear Hartree equation for the Coulomb potential. In addition, the convergence for the aforementioned equation is obtained in \cite{rodnianski2009quantum} with a rate of $N^{1/2}$ for the Coulomb potential. The convergence rate of $N^{-1}$ has been optimized in \cite{Chen2011} for the Coulomb potential, as well as for more singular potentials in \cite{Chen2018}.

This article is organized as follows. Brief introductions to the second quantization and Husimi measure are presented in Sections \ref{sec_second_bogo} and \ref{sec_def_husimi}, respectively. This is followed by the statement of our main theorem and proof strategy in Section \ref{sec_theorem_main}. Then, uniform estimates are given in Section \ref{sec_proof_main}, followed by the proof of the main theorem in Section \ref{sec_real_main_proof}. The estimates for the residual terms are covered in Section \ref{sec_est_residual}.

\section{Preliminaries}

\subsection{Second quantization}\label{sec_second_bogo}

In the study of large particle systems, we expect the operators to interact with different Hilbert spaces of the $N$-particle system by creating and annihilating particles. Therefore, to analyze a large particle system, it is convenient for us to build a `larger' Hilbert space that accompanies the aforementioned interactions, equipped with the norm $\norm{\ \cdot\ }$. In particular, for a large fermionic system, we consider the Fock space for fermions as
\begin{equation*}\label{def_fock}
		\mathcal{F}_a:= \C \bigoplus_{n \geq 1} L^2_a (\R^{3n}, (\mathrm{d}x)^{\otimes n}),
	\end{equation*}
where $L^2_a (\R^{3n}; (\mathrm{d}x)^{\otimes n})$ represents the $n$-fold antisymmetric tensor product of $L^2({\R^3})$. Moreover, the vacuum state is denoted as $\Omega=1\oplus 0 \oplus 0 \oplus \dots\in \mathcal{F}_a$.

For $f \in L^2(\R^3)$, the annihilation operator $a(f)$ and creation operator $a^*(f)$ acting on $\Psi = \bigoplus_{n \geq 0}\psi^{(n)} \in \mathcal{F}_a$ are defined by
\begin{align*}
		\big(a(f) \Psi \big)^{(n)} &:= \sqrt{n+1} \int \dd{x} \overline{f(x)} \psi^{(n+1)} (x, x_1, \dots, x_n),\\
		\big(a^*(f) \Psi  \big)^{(n)} &:= \frac{1}{\sqrt{n}} \sum_{j=1}^n f(x_j) \psi^{(n-1)} (x_1, \dots, {x}_{j-1}, {x}_{j+1} , \dots , x_n).
\end{align*}
Here $\Psi^{(n)}$ denotes the $n$-th particle sector of $\Psi \in \mathcal{F}_a$. 
Following the notations from \cite{benedikter2014mean}, we will use the operator valued distributions $a^*_x$ and $a_x$, to represent the creation and annihilation operators:
\begin{equation}\label{distrb_annil_creation_def}
		a^*(f) = \int \mathrm{d}x\ f(x) a^*_x, \quad a(f) = \int \mathrm{d}x\  \overline{f(x)}a_x.
	\end{equation}
Note that the operator-valued distribution $a^*_x$ formally creates a particle at position $x\in \R^3$, whilst the operator-value distribution $a_x$ annihilate a particle at $x$.

Furthermore, by the corresponding canonical anticommutation relations (CAR) in the fermionic system, we have that for any $f,g \in L^2(\R^3)$
\begin{equation}\label{def:CAR}
		\{ a(f), a^*(g)\} = \left< f,g\right>, \quad \{ a^*(f), a^*(g)\} = \{ a(f), a(g)\} = 0,
	\end{equation}
where $\{A, B\} := AB + BA$ is the anticommutator. Following from \eqref{def:CAR}, the CAR for operator kernels holds as follows:
\begin{equation}\label{eq:CAR}
	\{ a_x, a^*_y\} = \delta_{x=y}, \quad \{a^*_x, a^*_y\} = \{ a_x, a_y\} = 0.
\end{equation}

For any normalized $\Psi_{N,t} \in \mathcal{F}_a$, it is straightforward to show that
\begin{equation}\label{ineq_anil_create}
		\norm{a(f)\Psi_{N,t}}^2 \leq \norm{f}_{L^2}^2, \quad \norm{a^*(f)} = \norm{a(f)}
	\end{equation}
for any $f \in L^2(\R^3)$.\footnote{See Theorem 3.52 in \cite{Derezinski2009} for a more pedagogical approach to the annihilation and creation operator for the fermionic case.}

We extend the Hamilton operator appeared in \eqref{Schro_1} acting on $L^{2}_a(\mathbb{R}^{3N})$ to an operator acting on the Fock space $\mathcal{F}_a$ by $(\mathcal{H}_N\Psi)^{(n)} = \mathcal{H}_N^{(n)} \psi^{(n)}$ with
	\begin{equation*}
		\mathcal{H}_N^{(n)}  = \sum_{j=1}^{n} -\frac{\hbar^2}{2}\Delta_{x_j} + \frac{1}{2N}\sum_{i\neq j}^n V(x_i - x_j).
	\end{equation*}
	Then, we can write the Hamiltonian $\mathcal{H}_N$ 
	in terms of the operator-valued distributions $a_{x}$ and $a_{x}^{*}$ by
\begin{equation}\label{Fock_Hamil}
		\mathcal{H}_N = \frac{\hbar^2}{2} \int \mathrm{d}x\ \nabla_x a^*_x \nabla_x a_x + \frac{1}{2N} \iint \mathrm{d}x\mathrm{d}y\ V_N(x-y) a^*_x a^*_y a_y a_x.
	\end{equation}

In this article, we will consider only the following Schr\"odinger equation in Fock space:
\begin{equation}\label{Schrodinger_fock}
		\begin{cases}
			\displaystyle \mathrm{i} \hbar \partial_t \Psi_{N,t} =\mathcal{H}_N \Psi_{N,t}, \\
			\Psi_{N,0} = \Psi_N,
		\end{cases}
	\end{equation}
for all $\Psi_{N,t} \in \mathcal{F}_a$ and $\|\Psi_{N,t}\|=1$ for $t \in [0,T]$.

Next, we denote the number of particles operator and kinetic energy operator as
\begin{equation}\label{def_NumOp_KE}
		\mathcal{N} =  \int \mathrm{d}x\ a^*_x a_x \quad \text{and}\quad \mathcal{K} =  \frac{\hbar^2}{2}\int \mathrm{d}x\ \nabla_x a^*_x \nabla_x a_x,
	\end{equation}
respectively.

For any given $\Psi \in \mathcal{F}_a$ in the $n$-th sector, we can interpret the number of particles operator as
\begin{equation}\label{foobar_Num}
		(\mathcal{N}\Psi)^{(n)} = n \psi^{(n)},
	\end{equation}
where $\psi^{(n)} \in L_a^2(\R^{3n})$ for any $n \geq 1$. In the vacuum state, we have $\mathcal{N}\Omega = 0$. It is therefore straightforward to show that for $k \geq 1$,
\begin{equation*}
			\left< \Psi_{N,t}, \mathcal{N}^k \Psi_{N,t} \right>  = N(N-1)\cdots (N-k+1),
	\end{equation*}
for any normalized $\Psi_{N,t} \in \mathcal{F}_a$ and $t \geq 0$. Clearly, the relation between the number of particles operator and the $1$-particle reduced density matrix is given as
\begin{equation*}
		\left< \Psi_{N,t}, \mathcal{N} \Psi_{N,t} \right>  = \int \dd{w} \left< \Psi_{N,t},a^*_w a_w \Psi_{N,t} \right>   = \gamma_{N,t}^{(1)}(w;w),
	\end{equation*}
and observe $\Tr \gamma_{N,t}^{(1)} = N$.

\subsection{The Husimi measure}\label{sec_def_husimi}

We use the definition of the Husimi measure given in \cite{Fournais2018}. Let $f$ be any real-valued normalized function in Hilbert space; then, the coherent state is defined as
\begin{equation}\label{def_coherent}
		f^{\hbar}_{q, p} (y) := \hbar^{-\frac{3}{4}} f \left(\frac{y-q}{\sqrt{\hbar}} \right) e^{\frac{\mathrm{i}}{\hbar} p \cdot y }.
	\end{equation}
Then, the projection of coherent state is given by
\begin{equation*}\label{projection_f}
	\frac{1}{(2\pi \hbar)^3} \iint \mathrm{d}q \mathrm{d}p\,|f^{\hbar}_{q, p}\rangle\langle{f^{\hbar}_{q, p} }|  = \mathds{1}.
	\end{equation*}

For any $\Psi_{N,t} \in \mathcal{F}_a$, $1 \leq k \leq N$ and $t \geq 0$, the $k$-particle Husimi measure is defined as 
\begin{equation}\label{husimi_def_1}
	\begin{aligned}
		&m^{(k)}_{N,t} (q_1, p_1, \dots, q_k, p_k)\\
		&
		:= \dotsint (\mathrm{d}w\mathrm{d}u)^{\otimes k} \left( f^\hbar_{q,p}(w) \overline{f^\hbar_{q,p}(u)} \right)^{\otimes k}  \left< \Psi_{N,t},  a^*_{w_1} \cdots a^*_{w_k} a_{u_k} \cdots a_{u_1} \Psi_{N,t} \right>\\
		&
		= \dotsint (\mathrm{d}w\mathrm{d}u)^{\otimes k} \left( f^\hbar_{q,p}(w) \overline{f^\hbar_{q,p}(u)} \right)^{\otimes k} \gamma_{N,t}^{(k)}(u_1, \dots, u_k; w_1, \dots, w_k),
	\end{aligned}
\end{equation}
where we use the short notations
\[
(\mathrm{d}w\mathrm{d}u)^{\otimes k} := \mathrm{d}w_1 \mathrm{d}u_1 \cdots \mathrm{d}w_k \mathrm{d}u_k,\ \text{and}\
\left( f^\hbar_{q,p}(w) \overline{f^\hbar_{q,p}(u)} \right)^{\otimes k} := \prod_{j=1}^k f^\hbar_{q_j,p_j}(w_j) \overline{f^\hbar_{q_j,p_j}(u_j)}.
\]
 The Husimi measure defined in \eqref{husimi_def_1} measures how many particles, in particular fermions, are in the $k$-semiclassical boxes with a length scale of $\sqrt{\hbar}$ centered in its respective phase-space pairs, $(q_1, p_1), \dots, (q_k, p_k)$. 
 
 \begin{Remark}
 The Husimi measure \eqref{husimi_def_1} is a more generalized version of \eqref{gaussian_m_vs_W}. If $f$ is given by a Gaussian function, then the definitions of ${m}^{(k)}_{N,t}$ and $\mathfrak{m}^{(k)}_{N,t}$ coincide.
  \end{Remark}

Then, we observe that by using the operator kernels defined in \eqref{distrb_annil_creation_def}, the Husimi measure can be expressed by

The relation between the Husimi measure and the number of particles operator can be expressed as follows, for the 1-particle Husimi measure $m_{N,t}:=m^{(1)}_{N,t}$,
\begin{align*}
		\iint \mathrm{d}q\mathrm{d}p\ m_{N,t}(q,p) &=   \iint \mathrm{d}q\mathrm{d}p  \iint \mathrm{d}w_1 \mathrm{d}u_1\  f^\hbar_{q,p}(w_1) \gamma_{N,t}^{(1)}(w_1;u_1)\overline{f^\hbar_{q,p}(u_1)}\\
		& =  \hbar^{-\frac{3}{2}}  \int \mathrm{d}q  \iint \mathrm{d}w_1 \mathrm{d}u_1\ f \left(\frac{w_1-q_1}{\sqrt{\hbar}}\right) f \left(\frac{u_1-q_1}{\sqrt{\hbar}}\right) \left(\int \mathrm{d}p\  e^{\frac{\mathrm{i}}{\hbar}p \cdot (w_1 - u_1)} \right)  \gamma_{N,t}^{(1)}(w_1;u_1)\\
		&
		= (2\pi \hbar)^3 \hbar^{-\frac{3}{2}} \iint \mathrm{d}q_1\mathrm{d}w_1\ \left|f \left(\frac{w_1-q_1}{\sqrt{\hbar}}\right)\right|^2  \gamma_{N,t}^{(1)}(w_1;w_1)\\
		&
		=  (2\pi \hbar)^3 \int \mathrm{d}\tilde{q}\ |f(\tilde{q})|^2 \int\mathrm{d}w_1\  \gamma_{N,t}^{(1)}(w_1;w_1)\\
		&
		=  (2\pi)^3,
	\end{align*}
where we use the Dirac-delta $ \delta_{x}(y) := (2\pi \hbar)^{-3}\int  e^{\frac{\mathrm{i}}{\hbar}p \cdot (x - y)} \dd{p}$. Further properties of the Husimi measure are covered in Lemma \ref{prop_kHusimi}. Observe that if the initial data is described by Slater determinant as in \eqref{1_slater}, then the Husimi measure at initial time is 
	\begin{equation}\label{initial_husimi_slater}
		m_{N}^{\text{Slater}} (q,p) =   \sum_{j=1}^N \iint \mathrm{d}w_1 \mathrm{d}u_1\  f^\hbar_{q,p}(w_1) \overline{\mathrm{e}_j (w_1)}\mathrm{e}_j(u_1) \overline{f^\hbar_{q,p}(u_1)}.
	\end{equation}	

We are now ready to state the main theorem.

\section{Main result }\label{sec_theorem_main}
In this section, we provide our main result, proof strategies, and the \textit{a priori} estimates. The complete proof will be presented in Section \ref{sec_real_main_proof}. In the following, we denote $ \nabla_q f$ and $ \nabla_p f$ to be the gradients of $f$ with respect to the position and momentum variables respectively.
\begin{Theorem}\label{theorem_main}
Suppose that $V_N$ is the regularized Coulomb potential given in \eqref{def_V_N} with $\beta_{N} := N^{-\epsilon}$ and $0<\epsilon<\frac{1}{24}$ hold. For any fixed $T>0$, let $\Psi_{N,t} \in \mathcal{F}_a$, $t\in [0,T]$, be the solution to the Schr\"odinger equation \eqref{Schrodinger_fock} with the Slater determinant as the initial data. Let $m_{N,t}$ be the $1$-particle Husimi measure defined in \eqref{husimi_def_1}, where $f$ is a compact supported positive-valued function in $H^1(\R^3)$ with $\|f\|_{L^2}=1$. Moreover, let $m_{N}^{\text{Slater}}$ be the initial 1-particle Husimi measure with its $L^1$-weak limit $m_0$ and  there exists a constant $C>0$ independent of $N$ such that
\begin{equation}\label{inip2q}
	\iint \mathrm{d}q \mathrm{d}p\ (|p|^2 + |q|) m_{N}(q,p) \leq  C .
	\end{equation}
Then,  $m_{N,t}$ has a weak-$\star$ convergent subsequence in $L^\infty((0,T];L^1(\R^3 \times \R^3))$ with limit $m_{t}$, and $m_{t}$ is the solution of the Vlasov-Poisson equation with repulsive Coulomb potential,
\begin{equation}\label{eq_vlasov_poisson}
    	\begin{cases}
    		\partial_t m_{t}(q,p) + p \cdot \nabla_q m_t(q,p) = \nabla_q \big(|\cdot |^{-1} * \varrho_{t}\big)(q) \cdot \nabla_p  m_t(q,p),\\
    		m_{t}(q,p)\big|_{t= 0} = m_0(q,p),
    	\end{cases}
    \end{equation}
in the sense of distribution where $\varrho_{t} (q):= \int \dd{p}  m_{t} (q,p) $.
\end{Theorem}

\begin{Remark} 
	Since the total energy is conserved in this problem, the assumption of repulsive interacting potential is important to give uniform estimates both for kinetic energy and potential energy.\footnote{See Lemma \ref{kinetic_finite} below}. In fact, the result in Theorem \ref{theorem_main} holds also for attractive singular potential if the kinetic energy can be bounded uniformly in $N$.
\end{Remark}

\begin{Remark} It is proven in Proposition \ref{2nd_moment_finite} that the first moment of the Husimi measure $m_{N,t}$ is uniformly bounded. Therefore, by Theorem 7.12 in \cite{Villani2003}, the convergence stated in theorem also holds in terms of the $1$-Wasserstein metric.\footnote{The $1$-Wasserstein metric is defined as
$
W_1 (\mu, \nu) := \max_{\pi \in \Pi (\mu,\nu)} \int |x-y|\ \mathrm{d}\pi(x,y),
$
where $\mu$ and $\nu$ are probability measures and $\Pi(\mu,\nu)$ the set of all probability measures with marginals $\mu$ and $\nu$.\cite{Villani2003}}
\end{Remark}

\begin{Remark}\label{rem_gosle}
In \cite{Golse2017}, the rate of convergence from Schr\"odinger to the Vlasov equation in the pseudometric is obtained for the interaction potential $V \in C^{1,1}$. In addition, the authors commented that their result can be extended for the truncated Coulomb interaction, but with order higher than $C/\sqrt{\ln N}$ for some constant $C >0$. In Theorem \ref{theorem_main}, the mollification of the Coulomb interaction can be handled with polynomial truncation.
\end{Remark}

\begin{Remark} 
	The global existence of classical solution to the Vlasov-Poisson equation in 3-dimension is proven in \cite{pfaffelmoser1992global} and \cite{lions1991propagation} for a general class of initial data. The uniqueness of the solution is proven in \cite{lions1991propagation} for initial datum with strong moment conditions and integrability. In \cite{Loeper2006}, the uniqueness of the solution is also proven for bounded macroscopic density. Furthermore, the global existence of weak solutions is provided in \cite{Arsenextquotesingleev1975} for bounded initial data and kinetic energy. The result is then relaxed to only $L^p$-bound for $p>1$ in \cite{gilbarg2015elliptic}. Result on existence with symmetric initial data is proven in \cite{batt1977global,Dobrushin1979,schaeffer1987global}. For other results, we refer to the works given in \cite{ambrosio2014lagrangian,bohun2016lagrangian,horst1981classical} to list a few.
	
\end{Remark}

\subsection{Proof strategies}\label{sec_proof_strat}
From \cite[Proposition 2.1]{Chen2021}, we obtain the following equation from the Schrödinger equation given \eqref{Schro_1}, i.e.,
\begin{equation} 
	\begin{aligned}
	&\p_t m_{N,t}(q,p) + p \cdot \nabla_{q} m_{N,t}(q,p) - \nabla_q\cdot\left(\hbar \Im \left< \nabla_{q} a (f^\hbar_{q,p}) \psi_{N,t}, a (f^\hbar_{q,p}) \psi_{N,t} \right>\right)\\
	&= \frac{1}{(2\pi)^3}\nabla_p\cdot\iint \mathrm{d}w_1 \mathrm{d}u_1 \iint \mathrm{d}w_2\mathrm{d}u_2 \iint \mathrm{d}q_2 \mathrm{d}p_2    \left(  f^\hbar_{q,p}(w)  \overline{f^\hbar_{q,p}(u)} \right)^{\otimes 2}\\
	& \hspace{1.5cm}\int_0^1 \mathrm{d}s\ \nabla V_N\big( su_1+ (1-s)w_1-w_2 \big)  \gamma_{N,t}^{(2)}(u_1,u_2;w_1,w_2),
	\end{aligned}
	\end{equation}
where we denote
$$\left(  f^\hbar_{q,p}(w)  \overline{f^\hbar_{q,p}(u)} \right)^{\otimes 2} :=  f^\hbar_{q,p}(w_1)  \overline{f^\hbar_{q,p}(u_1)}  f^\hbar_{q_2,p_2}(w_2)  \overline{f^\hbar_{q_2,p_2}(u_2)}.$$
In particular, this can be rewritten into the Vlasov equation with remainder terms, i.e.,
\begin{equation} \label{BBGKY_k1}
		\begin{aligned}
			&\p_t m_{N,t}(q,p) + p \cdot \nabla_{q} m_{N,t}(q,p)\\
			&= \frac{1}{(2\pi)^3}\nabla_{p} \cdot  \int \dd{q_2} \nabla V_N(q- q_2) \varrho_{N,t}(q_2) m_{N,t}(q,p)   
			+ \nabla_{q}\cdot \tilde{\mathcal{R}} +\nabla_{p}\cdot \mathcal{R},
		\end{aligned}
	\end{equation}
where $\varrho_{N,t}(q):= \int \dd{p} m_{N,t}(q,p)$, $\tilde{\mathcal{R}}$ and $\mathcal{R}=\mathcal{R}_1+\mathcal{R}_2$ are given by
\begin{equation}\label{bbgky_remainder_1}
		\begin{aligned}
		\tilde{\mathcal{R}} :=  &  \hbar \Im \left< \nabla_{q} a (f^\hbar_{q,p}) \psi_{N,t}, a (f^\hbar_{q,p}) \psi_{N,t} \right>,
		\\
		{\mathcal{R}_1} :=& \frac{1}{(2\pi)^3}  \iint \mathrm{d}w_1 \mathrm{d}u_1 \iint \mathrm{d}w_2\mathrm{d}u_2 \iint \mathrm{d}q_2 \mathrm{d}p_2    \left(  f^\hbar_{q,p}(w)  \overline{f^\hbar_{q,p}(u)} \right)^{\otimes 2}\\
		&
		\qquad \bigg[\int_0^1 \mathrm{d}s\ \nabla V_N\big( su_1+ (1-s)w_1-w_2 \big)   - \nabla V_N(q-q_2) \bigg]\gamma_{N,t}^{(2)}(u_1,u_2;w_1,w_2)  ,
		\\
		{\mathcal{R}_2} :=& 	\frac{1}{(2\pi)^3}  \iint \mathrm{d}w_1 \mathrm{d}u_1 \iint \mathrm{d}w_2\mathrm{d}u_2 \iint \mathrm{d}q_2 \mathrm{d}p_2    \left(  f^\hbar_{q,p}(w)  \overline{f^\hbar_{q,p}(u)} \right)^{\otimes 2}\\
		& 
		\qquad  \nabla V_N(q-q_2) \bigg[ \gamma_{N,t}^{(2)}(u_1,u_2;w_1,w_2) -\gamma_{N,t}^{(1)}(u_1;w_1)\gamma_{N,t}^{(1)}(u_2;w_2)\bigg].
		\end{aligned}
	\end{equation}
The main contribution of this article is to rigorously prove the limit $N\to\infty$ from \eqref{BBGKY_k1} to the Vlasov-Poisson equation \eqref{eq_vlasov_poisson} in the sense of distribution.

First, from the uniform estimate of the kinetic energy shown in Lemma \ref{kinetic_finite}, we prove in Proposition \ref{2nd_moment_finite} the uniform estimate for the moments of Husimi measure. Additionally, because the Husimi measure belongs to $ L^\infty([0,T];L^1(\R^3) \cap L^\infty(\R^3))$ (see Lemma \ref{prop_kHusimi}), we obtain directly the weak compactness of the two linear terms on the left-hand side of \eqref{BBGKY_k1} by the Dunford-Pettis theorem.\footnote{See Proposition \ref{propweaklimit}.}

For the quadratic term on the right-hand side of \eqref{BBGKY_k1}, the classical Thomas-Fermi theory gives that $\varrho_{N,t} \in L^\infty([0,T]; L^{5/3}(\R^3))$. With the \textit{a priori} estimate obtained in Section \ref{sec_proof_main}, the Aubin-Lions compact embedding theorem shows the strong compactness of $\nabla V_N  * \varrho_{N,t}$.

The estimate for the remainder term $\tilde{\mathcal{R}}$ is provided in \cite[Proposition 2.4]{Chen2021}. Thus, the main work of this paper is dealing with the challenging term $\mathcal{R}$. Unlike the BBGKY hierarchy used in \cite{Chen2021}, where the remainder term contains only the difference between the $2$-particle density matrices, we write the term $\mathcal{R}$ as a combination of the semiclassical and mean-field terms as $\mathcal{R}_1$ and $\mathcal{R}_2$, respectively.\footnote{See \eqref{remainder_total} for the full structure.} Thus, the factorization effect can be directly obtained from $\mathcal{R}_2$ instead of using the method of the BBGKY hierarchy.

The estimates for $\mathcal{R}_1$ and $\mathcal{R}_2$ are shown in Proposition \ref{est_semi_R} and Proposition \ref{est_meanfield_R} respectively, in which we utilize the estimates of the `cutoff' number operator and momentum oscillation presented in Lemma \ref{N_hbar} and Lemma \ref{estimate_oscillation}, to control the growth of the Lipschitz constant $V_N$, which is of order $\beta_{N}^{-2}$.

\subsection{\textit{A priori} estimates}\label{sec_proof_main}

We present in this subsection a sequence of estimates that is used repeatedly in the proof.

First, we cite the following properties of $k$-particle Husimi measures from
(or \cite[Lemma 2.2]{Chen2021} for the time dependent version).

\begin{Lemma} \label{prop_kHusimi}
Suppose that $\Psi_{N,t} \in \mathcal{F}_a$ is normalized for any $t \geq 0$. Then, the following properties hold true for $m^{(k)}_{N,t}$,
\begin{enumerate}
\item  $m^{(k)}_{N,t}(q,p,\dots,q_k,p_k)$ is symmetric,
\item  $\frac{1}{(2\pi)^{3k}} \dotsint (\mathrm{d}q\mathrm{d}p)^{\otimes k} m^{(k)}_{N,t}(q,p,\dots,q_k,p_k) = \frac{N(N-1)\cdots (N-k+1)}{N^k}$,
\item $\frac{1}{(2\pi \hbar)^{3}} \iint \mathrm{d}q_k  \mathrm{d}p_k\ m^{(k)}_{N,t}(q,p,\dots,q_k,p_k) = (N-k+1) m^{(k-1)}_{N,t}(q,p,\dots,q_{k-1},p_{k-1}) $,
\item $ 0 \leq  m^{(k)}_{N,t}(q,p,\dots,q_k,p_k) \leq 1$ a.e.,
\end{enumerate}
where $1 \leq k \leq N$.
\end{Lemma}

Then, due to the conservation of energy and the repulsive effect of the Coulomb force, we obtain the following estimate for the kinetic energy.

\begin{Lemma}\label{kinetic_finite}
Assuming that $V_N(x) \geq 0$ and the initial total energy is bounded in the sense that $\frac{1}{N}\langle\Psi_{N},\mathcal{H}_{N}\Psi_{N}\rangle\leq C$, then there exists a constant $C>0$ independent of $N$ such that
\begin{equation}\label{k_kinetic_bounded}
		\left<\Psi_{N,t}, \frac{\mathcal{K}}{N} \Psi_{N,t}\right> \leq C.
	\end{equation}
\end{Lemma}
\begin{proof}
We define the operator
\[
	\mathcal{V}_N := \frac{1}{N} \iint \mathrm{d}x \mathrm{d}y\ V_N(x-y) a^*_x a^*_y a_y a_x.
	\]
Since $V_N\geq0$, we have $\langle\Psi_{N,t},\mathcal{V}_{N}\Psi_{N,t}\rangle\geq0$. Then
\begin{align*}
		\langle\Psi_{N,t},\mathcal{H}_{N}\Psi_{N,t}\rangle=\langle\Psi_{N,t},\mathcal{K}\Psi_{N,t}\rangle+\langle\Psi_{N,t},\mathcal{V}_{N}\Psi_{N,t}\rangle,
	\end{align*}
implies
\[
	0\leq\langle\Psi_{N,t},\mathcal{K}\Psi_{N,t}\rangle\leq\langle\Psi_{N,t},\mathcal{H}_{N}\Psi_{N,t}\rangle.
	\]
Hence,
\[
	\frac{1}{N}\langle\Psi_{N,t},\mathcal{K}\Psi_{N,t}\rangle \leq \frac{1}{N}\langle\Psi_{N,t},\mathcal{H}_{N}\Psi_{N,t}\rangle = \frac{1}{N}\langle\Psi_{N},\mathcal{H}_{N}\Psi_{N}\rangle\leq C.
	\]
\end{proof}

Consequently, the moment estimate of the Husimi measure is obtained directly from the uniform bound in Lemma \ref{kinetic_finite}.

\begin{Proposition} \label{2nd_moment_finite} For $t \geq 0$, we have the following finite moments:
\begin{equation}\label{2nd_moment_finite_0}
				\iint \mathrm{d}q \mathrm{d}p\ (|{q}| + |{p}|^2) m_{N,t}(q,p) \leq C(1+t),
		\end{equation}
where $C>0$ is a constant that depends on initial data $\iint \mathrm{d}q \mathrm{d}p\ (|q| + |p|^2) m_{N}(q, p) $.
\end{Proposition}

\begin{proof}
First, from equation (2.16) in \cite{Chen2021}, we obtain that
\begin{equation}
		\begin{aligned}
			\left<\psi_{N,t},\frac{\mathcal{K}}{N}\psi_{N,t} \right> =  \frac{1}{(2\pi)^3}\iint \mathrm{d}q\mathrm{d}p\ |p|^2  m_{N,t}(q, p)
			+ \hbar \int \mathrm{d}q \left| \nabla f\left(q\right) \right|^2,
		\end{aligned}
	\end{equation}
which implies that
\begin{equation}\label{moment_momentum_1}
		\begin{aligned}
			\frac{1}{(2\pi)^3}\iint \mathrm{d}q\mathrm{d}p\ |p|^2  m_{N,t}(q, p) 
			&\leq  \left<\psi_{N,t},\frac{\mathcal{K}}{N} \psi_{N,t} \right> \leq C,
		\end{aligned}
	\end{equation}
where we use Lemma \ref{kinetic_finite} in the last inequality.

Then, for the moment with respect to $q$, we obtain from \eqref{BBGKY_k1} that
\begin{equation}\label{q_moment}
		\begin{aligned}
			&\dv{t} \iint \mathrm{d}q \mathrm{d}p\ |q| m_{N,t}(q,p) =  \iint \mathrm{d}q \mathrm{d}p\ |q|\p_t m_{N,t}(q,p)\\
			= &\iint \mathrm{d}q \mathrm{d}p\ |q| \bigg( - p \cdot \nabla_{q} m_{N,t}(q,p) + \frac{1}{(2\pi)^3} \nabla_{p} \cdot \iint \mathrm{d}w_1\mathrm{d}u_1 \iint \mathrm{d}w_1\mathrm{d}u_2 \iint \mathrm{d}q_2 \mathrm{d}p_2 \int_0^1 \mathrm{d}s\\
			&  \nabla V\big(su_2+(1-s)w_1 - w_2 \big) f_{q,p}^\hbar (w_1) \overline{f_{q,p}^\hbar (u_1)} f_{q_2,p_2}^\hbar (w_2) \overline{f_{q_2,p_2}^\hbar (u_2)}  \left<a_{w_2}a_{w_1} \Psi_{N,t}, a_{u_2} a_{u_1} \Psi_{N,t} \right> + \nabla_{q} \cdot \tilde{\mathcal{R}} \bigg).
		\end{aligned}
	\end{equation}

By applying the divergence theorem first with respect to $p$ and then with respect to $q$ in \eqref{q_moment}, we obtain
\begin{align*}
			\dv{t} \iint \mathrm{d}q \mathrm{d}p\ |q| m_{N,t}(q,p) & =  \iint \mathrm{d}q \mathrm{d}p\ \frac{q}{|q|} \cdot p \  m_{N,t}(q,p)\\
			& 
			\leq  \iint \mathrm{d}q \mathrm{d}p\ (1+|p|^2) \cdot  m_{N,t}(q,p),
		\end{align*}
where we use Young's product inequality. Finally, taking the integral over $t$, we obtain the desired result.
\end{proof}

\subsection{Proof of Theorem \ref{theorem_main}}\label{sec_real_main_proof}

First, denoting $\varrho_{N,t}(q) := \int m_{N,t}(q,p) \dd{p}$, recall the Vlasov equation
\begin{equation} \label{BBGKY_k1_compact}
		\begin{aligned}
			&\p_t m_{N,t}(q,p) + p \cdot \nabla_{q} m_{N,t}(q,p)\\
			&= \frac{1}{(2\pi)^3}\nabla_{p} \cdot  \int \dd{q_2} \nabla V_N(q- q_2) \varrho_{N,t}(q_2) m_{N,t}(q,p)   
			+ \nabla_{q}\cdot \tilde{\mathcal{R}} +\nabla_{p}\cdot {\mathcal{R}}\\
			&
			= \frac{1}{(2\pi)^3} (\nabla V_N * \varrho_{N,t}) (q)\cdot \nabla_p m_{N,t} (q,p)+  \nabla_q \cdot \tilde{\mathcal{R}} + \nabla_p \cdot  {\mathcal{R}},
		\end{aligned}
	\end{equation}
with
\begin{equation}\label{214}
		\begin{aligned}
			\tilde{\mathcal{R}} :=  &  \hbar \Im \left< \nabla_{q} a (f^\hbar_{q,p}) \Psi_{N,t}, a (f^\hbar_{q,p}) \Psi_{N,t} \right>,
			\\
			{\mathcal{R}}:=& (2\pi)^3  \iint \mathrm{d}w_1 \mathrm{d}u_1 \iint \mathrm{d}w_2\mathrm{d}u_2 \iint \mathrm{d}q_2 \mathrm{d}p_2    \left(  f^\hbar_{q,p}(w)  \overline{f^\hbar_{q,p}(u)} \right)^{\otimes 2}\\
			&
			\qquad \bigg[\int_0^1 \mathrm{d}s\ \nabla V_N\big( su_1+ (1-s)w_1-w_2 \big) \gamma_{N,t}^{(2)}(u_1,u_2;w_1,w_2)  \\
			& 
			\qquad   - \nabla V_N(q-q_2)\gamma_{N,t}^{(1)}(u_1;w_1)\gamma_{N,t}^{(1)}(u_2;w_2)\bigg].
		\end{aligned}
	\end{equation}
The main task is now reduced to taking limits in \eqref{214}. In fact, Section \ref{sec_est_residual} is devoted to deriving the estimates for the residuals. As a summary, it is proven in Section \ref{sec_est_residual} that for $\varphi,\phi\in C^\infty_0 (\R^3)$, there exists a positive constant ${K}$ such that
\begin{equation}\label{est_RtildeR_sum}
		\begin{aligned}
			\left|\iint \mathrm{d}q \mathrm{d}p\ \varphi(q)\phi(p) \nabla_{q}\cdot \tilde{\mathcal{R}}(q,p) \right| &\leq {K} \hbar^{\frac{1}{2}-\delta},\\
			\left|\iint \mathrm{d}q \mathrm{d}p\ \varphi(q)\phi(p)  \nabla_{p} \cdot \mathcal{R}(q,p)  \right| & \leq  {K} \big( \hbar^{\frac{1}{4}(6\alpha_1 - 5)-2\delta} +  \hbar^{\frac{3}{2}(\alpha_2 - \frac{1}{2}) -2\delta}\big),
		\end{aligned}
	\end{equation}
where $\frac{5}{6} < \alpha_1 < 1$, $\frac{1}{2} < \alpha_2 < 1$ and $0 < \delta \ll 1$. The estimates in \eqref{est_RtildeR_sum} show that the residual terms converge to zero in the sense of distribution.

Next, we have the following result on weak 
convergent in $L^1$:

\begin{Proposition}[Proposition 2.7 of \cite{Chen2021}]\label{lem_dunford}
	\label{propweaklimit}
Let $\{m_{N,t}\}_{N\in \N}$ be the $1$-particle Husimi measure; then, there exists a subsequence $\{m_{N_j,t}\}_{j\in \N}$ that converges weakly in $L^1(\R^{3} \times \R^3 )$ to a function $ (2\pi)^{3}m_{t}$; i.e., for all $ \Phi\in L^\infty (\R^{3} \times \R^3)$, it holds that
\[
	\frac{1}{(2\pi)^{3}}\iint \mathrm{d}q\mathrm{d}p\  m_{N_j,t} (q,p)\Phi(q,p) \rightarrow \iint \mathrm{d}q\mathrm{d}p\   m_{t} (q,p)\Phi (q,p),
	\]
as $j \to \infty$.
\end{Proposition}
\begin{Remark}
The proof for Lemma \ref{lem_dunford} is obtained by proving its uniform integrability and employing the Dunford-Pettis theorem for $L^1$ compactness.
\end{Remark}

Furthermore, to prove the convergence of the nonlinear term $(\nabla V_N * \rho_N) \cdot \nabla_p m_{N,t}$, we first show the strong convergence of $\nabla V_N * \rho_N$.

\begin{Lemma}\label{lem_compact_argument}
Let $V_N$ be defined as \eqref{def_V_N}. Then for $t \in [0,\infty)$ there exists constant $C>0$ independent on $N$ such that
\begin{align}
		\norm{\nabla V_N * \varrho_{N,t}}_{L^{\infty} ([0,\infty); W^{1,  \frac{5}{3}}(\R^3))}\leq C, \label{compact_1}\\ 
		\norm{\p_t (\nabla V_N * \varrho_{N,t})}_{L^{\infty} ([0,\infty); W^{-1,\frac{15}{7}}(\R^3))}\leq C. \label{compact_2}
	\end{align}
\end{Lemma}

\begin{proof}
From Lemma \ref{prop_kHusimi} and Proposition \ref{2nd_moment_finite}, one finds that $m_{N,t}$ is uniformly bounded in $L^{\infty} ([0,\infty); L^1(\R^{3} \times \R^3 ))\cap L^{\infty} ([0,\infty); L^\infty(\R^{3} \times \R^3 ))$ and $|p|^2m_{N,t}(q,p)$ uniformly in $L^{\infty} ([0,\infty); L^1(\R^{3} \times \R^3 ))$ respectively. As a consequence, it holds that
	$$
	\norm{\varrho_{N,t}}_{L^{\infty} ([0,\infty); L^{\frac{5}{3}}(\R^3))}\leq C.
	$$ 
	Thus, $V_N * \varrho_{N,t}=V*\mathcal{G}_{\beta_N}* \varrho_{N,t}$ is uniformly bounded in $L^{\infty} ([0,\infty); W^{2,  \frac{5}{3}}(\R^3))$ due to the fact that $V$ is the fundamental solution of the Poisson equation and 
	$$\norm{\mathcal{G}_{\beta_N}* \varrho_{N,t}}_{L^{\infty} ([0,\infty); L^{\frac{5}{3}}(\R^3))}\leq \norm{\mathcal{G}_{\beta_N}}_{L^1(\R^3)}\cdot \norm{\varrho_{N,t}}_{L^{\infty} ([0,\infty); L^{\frac{5}{3}}(\R^3))}. 
	$$
This implies the result \eqref{compact_1} directly.

To prove \eqref{compact_2}, recall again the transport equation for $m_{N,t}$
\begin{equation}
		\partial_t m_{N,t} + p\cdot \nabla_q m_{N,t} - \frac{1}{(2\pi)^3} (\nabla V_N * \varrho_{N,t}) \cdot \nabla_p m_{N,t} =  \nabla_q \cdot \tilde{\mathcal{R}} + \nabla_p \cdot  {\mathcal{R}},
	\end{equation}
where $\varrho_{N,t}(q) := \int m_{N,t}(q,p) \dd{p}$. Taking the integral with respect to $p$,
\begin{align*}
		&\partial_t \int \dd{p}  m_{N,t}(q,p)  + \nabla_q \cdot \int \dd{p} p\ m_{N,t}(q,p) = \nabla_q \cdot \int \dd{p}  \tilde{\mathcal{R}} .
	\end{align*}
Next, by taking the convolution with $\nabla V_N$, we obtain
\begin{equation}\label{vlasov_conserve}
		\partial_t \big(\nabla V_N * \varrho_{N,t}\big) + \nabla_q \cdot (\nabla V_N \otimes_* J_{N,t}) = \nabla_q \cdot\left( \nabla V_N \otimes_* \int \dd{p} \tilde{\mathcal{R}}\right),
	\end{equation}
where $J_{N,t}(q) :=  \int \dd{p}p\ m_{N,t}(q,p) $, $(u\otimes_*v)_{ij}=u_i*v_j$ for $u,v\in\R^3$. Then, we observe that
\begin{equation}\label{compact_0}
			\left|\int \dd{p} p \ m_{N,t} (q,p)\right|  \\
			\leq \left[\int \dd{p} |p|^2 m_{N,t}\right]^\frac{1}{2} \left[\int \dd{p} m_{N,t}\right]^\frac{1}{2}
			=  \left[\int \dd{p} |p|^2 m_{N,t}\right]^\frac{1}{2} \varrho_{N,t}^\frac{1}{2}.
	\end{equation}
Therefore, we have
\begin{align*}
		\int \dd{q} \left|J_{N,t}(q)\right|^\frac{5}{4} & = \int \dd{q} \left|\int \dd{p} p \ m_{N,t} (q,p)\right|^\frac{5}{4} 
		\leq  \left[\iint \mathrm{d}q\mathrm{d}p\ |p|^2 m_{N,t}\right]^\frac{5}{8} \left[ \int \dd{q}\varrho_{N,t}^\frac{5}{3} \right]^\frac{3}{8} 
		\leq C,
	\end{align*}
where we use Proposition \ref{2nd_moment_finite} in the last inequality, yielding that $J_{N,t}$ is uniformly bounded in $L^{\infty}\big([0,\infty);L^{\frac{5}{4}}(\R^3)\big)$. Then, for any test function $\varphi \in L^{\frac{15}{8}}(\R^3)$, we obtain for a.e. $t\geq 0$ that
\begin{align*}
	&	\int \dd{q} |\varphi(q)| \left|\int \dd{q_2} \nabla V_N (q-q_2) J_{N,t}(q_2)\right| \leq \int \mathrm{d}q |\varphi(q)| \left|\int\mathrm{d}q_2 \nabla V(q-q_2) \mathcal{G}_{\beta_N}*J_{N,t}(q_2)\right|\\
		&
			\qquad \leq \iint \mathrm{d}q \mathrm{d}q_2  |\varphi(q)| |\nabla V(q-q_2)| |\mathcal{G}_{\beta_N}*J_{N,t}(q_2)|\leq  C \norm{\varphi}_{L^{\frac{15}{8}}} \norm{\mathcal{G}_{\beta_N}*J_{N,t}}_{L^{\frac{5}{4}}}  \\
		& 
			\qquad \leq C \norm{\varphi}_{L^{\frac{15}{8}}} \norm{\mathcal{G}_{\beta_N}}_{L^1} \norm{J_{N,t}}_{L^{\frac{5}{4}}}\leq C \norm{\varphi}_{L^{\frac{15}{8}}},
	\end{align*}
where we use the Hardy-Littlewood-Sobolev inequality in the third inequality.
This implies that, by using dual formulation of $L^p$ norms and taking the supremum in $\varphi \in L^{\frac{15}{8}}$, one obtains that $\nabla V_N * J_{N,t}$ is uniformly bounded in  $L^{\infty}([0,\infty); L^{\frac{15}{7}}(\R^3) ) $.

Therefore, focusing on the estimate of $\tilde{\mathcal{R}}$, we have
\begin{align*}
		\left| \int  \mathrm{d}p \  \tilde{\mathcal{R}}\right| & 
		\leq  \hbar  \int  \mathrm{d}p \  \Big|\left< \nabla_{q} a (f^\hbar_{q,p}) \Psi_{N,t},  a (f^\hbar_{q,p}) \Psi_{N,t} \right>  \Big| \\
		&
		\leq \hbar  \int \mathrm{d}p \norm{\nabla_{q} a (f^\hbar_{q,p}) \Psi_{N,t}}\norm{ a (f^\hbar_{q,p}) \Psi_{N,t} }\\
		&
		\leq   \left[ \hbar^2 \int \mathrm{d}p  \left< \nabla_{q} a (f^\hbar_{q,p}) \Psi_{N,t},   \nabla_{q} a (f^\hbar_{q,p}) \Psi_{N,t} \right> \right]^\frac{1}{2} \left[\int  \mathrm{d}p\ m_{N,t}(q,p) \right]^\frac{1}{2}\\
		&
		= \left[ \hbar^2 \int \mathrm{d}p  \left< \nabla_{q} a (f^\hbar_{q,p}) \Psi_{N,t},   \nabla_{q} a (f^\hbar_{q,p}) \Psi_{N,t} \right> \right]^\frac{1}{2} \varrho_{N,t}^\frac{1}{2}.
	\end{align*}
Note that since it holds that
\begin{align*}
		&  \hbar^2 \iint \mathrm{d}q \mathrm{d}p  \left< \nabla_{q} a (f^\hbar_{q,p}) \Psi_{N,t},   \nabla_{q} a (f^\hbar_{q,p}) \Psi_{N,t} \right> \\
		&
		= \hbar^{\frac{1}{2}} \iint \mathrm{d}q \mathrm{d}p  \iint  \mathrm{d}w \mathrm{d}u\   \nabla_{q} f\left(\frac{w-q}{\sqrt{\hbar}}\right)\nabla_{q} f\left(\frac{u-q}{\sqrt{\hbar}}\right) e^{\frac{\mathrm{i}}{\hbar}p \cdot (w-u)} \left< \Psi_{N,t},  a^*_w a_u \Psi_{N,t} \right>\\
		&
		= \hbar^{\frac{1}{2}+3} \iint \mathrm{d}q \mathrm{d}w\ \hbar^{-1}\left|\nabla f\left(\frac{w-q}{\sqrt{\hbar}}\right) \right|^2 \left< \Psi_{N,t},  a^*_w a_w \Psi_{N,t} \right>\\&
		= \hbar^{4} \int \dd{\tilde{q}} \left|\nabla f(\tilde{q}) \right|^2 \left< \Psi_{N,t}, \mathcal{N} \Psi_{N,t} \right> \\
		&
		\leq \hbar \norm{\nabla f}_2^2,
	\end{align*}
this implies
\begin{align*}
		\int \dd{q} \left| \int  \mathrm{d}p \  \tilde{\mathcal{R}} \right|^\frac{5}{4} 
		\leq \left[ \hbar^2 \iint \mathrm{d}q \mathrm{d}p\ \left< \nabla_{q} a (f^\hbar_{q,p}) \Psi_{N,t},   \nabla_{q} a (f^\hbar_{q,p}) \Psi_{N,t} \right> \right]^\frac{5}{8} \left(\int \dd{q}\varrho_{N,t}^{\frac{5}{3}}\right)^\frac{3}{8}
		\leq \hbar^{\frac{5}{4}} C.
	\end{align*}
Repeating the calculation in \eqref{compact_0}, we have that $\nabla V_N * \int \dd{p} \tilde{\mathcal{R}}$ is uniformly bounded in $L^{\infty}([0,\infty); L^{\frac{15}{7}}(\R^3))$, which implies that $\nabla V_N *  \int \dd{p} (\nabla_q \cdot \tilde{\mathcal{R}} + \nabla_p \cdot  {\mathcal{R}})$ is uniformly bounded in $ L^{\infty}([0,\infty); L^{\frac{15}{7}}(\R^3))$. Thus, from \eqref{vlasov_conserve}, we have that there exists a $C$ which is independent of $N$ such that
\[
\norm{\p_t (\nabla V_N * \varrho_{N,t})}_{L^{\infty} ([0,\infty); W^{-1,\frac{15}{7}}(\R^3) )}\leq C.
	\]
This completes the proof for Lemma \ref{lem_compact_argument}.
\end{proof}

Finally, we conclude the proof of main theorem with the following compactness argument.

\subsubsection{Compactness argument}

As in Section \ref{sec_proof_strat}, the weak convergence of the linear terms in the Vlasov equation is obtained from Proposition \ref{lem_dunford}. The following discussion is focused on the nonlinear term. Without loss of generality, assume that $\Phi(q,p) =\varphi(q) \phi(p)$ for any test functions $\varphi$, $\phi \in C^\infty_0 (\R^3)$, and let the sphere $B_\ell$ with radius $\ell>0$ be the support of $\varphi$. Due to the Sobolev's embedding theorem, we have
\[
	W^{1,\frac{5}{3}} (B_\ell)\doublehookrightarrow L^{r} (B_\ell) \hookrightarrow  W^{-1,\frac{15}{7}}(B_\ell),
	\]
where $\frac{5}{4} \leq r < \frac{15}{4}$ and $\doublehookrightarrow$ means the compact embedding. Recall the results in Lemma \ref{lem_compact_argument}, we have
$$
\norm{(\nabla V_N * \varrho_{N,t})}_{L^{\infty} ([0,\infty); W^{1,  \frac{5}{3}}(\R^3))} + 
\norm{\p_t (\nabla V_N * \varrho_{N,t})}_{L^{\infty} ([0,\infty); W^{-1,\frac{15}{7}}(\R^3))}\leq C.
$$
Then, by Aubin-Lion lemma, we obtain that there exists a subsequence denoted also by $\big(\nabla V_N * \varrho_{N,t}\big)_{N \in \N}$, and $h \in L^{\infty}([0,T]; L^r (B_\ell))$ such that, as $N \to \infty$, we have
\begin{equation}\label{compact_converg_0}
	 \nabla V_N * \varrho_{N,t} \rightarrow h \quad\text{in} \quad L^{\infty}([0,T]; L^r ( \R^3)),
	\end{equation}
where $\frac{5}{4} \leq r < \frac{15}{4}$. The weak star convergence of $\varrho_{N,t} \rightharpoonup^*\varrho_{t}$ in $L^\infty ((0,T); L^\frac{5}{3}(\R^3))$, where $\varrho_{t}(q) := \int m_{t}(q,p) \dd{p}$, and the definition of $V_N$ in \eqref{def_V_N} imply that the limit function $h$ coincides with $\nabla V*\varrho_t$ a.e. in $B_\ell$.

Now, to show the convergence to the Vlasov-Poisson equation, we first compute
\begin{align*}
		& \left|\int_0^T \dd{t} \iint \mathrm{d}q \mathrm{d}p\   \varphi(q) \nabla_{p}\phi(p) \cdot \left[(\nabla V_N * \varrho_{N,t})(q)  m_{N,t}(q,p) - (\nabla V * \varrho_{t})  m_{t}(q,p) \right]  \right|\\
		&
		= \bigg|\int_0^T \dd{t} \iint \mathrm{d}q \mathrm{d}p\  \varphi(q) \nabla_{p}\phi(p) \cdot \left[(\nabla V_N * \varrho_{N,t})(q) - (\nabla V * \varrho_{t}) (q) \right] m_{N,t}(q,p)  \\
		&
		\qquad +  \int_0^T \dd{t} \iint \mathrm{d}q \mathrm{d}p\  \varphi(q) \nabla_{p}\phi(p) \cdot  (\nabla V * \varrho_{t})(q)  \left[m_{N,t}(q,p) - m_{t}(q,p) \right] \bigg|\\
		&
		\leq \bigg|\int_0^T \dd{t} \iint \mathrm{d}q \mathrm{d}p\  \varphi(q) \nabla_{p}\phi(p)\cdot \left[(\nabla V_N * \varrho_{N,t}) (q)- (\nabla V * \varrho_{t}) (q) \right] m_{N,t}(q,p) \bigg| \\
		&
		\qquad +  \bigg|\int_0^T \dd{t} \iint \mathrm{d}q \mathrm{d}p\   \varphi(q) \nabla_{p}\phi(p) \cdot  (\nabla V * \varrho_{t})(q)  \left[m_{N,t}(q,p) - m_{t}(q,p) \right] \bigg|
		= : \mathcal{A}_1 + \mathcal{A}_2.
	\end{align*}

Let us focus on the first term.
\begin{align*}
	 \mathcal{A}_1 & = \bigg|\int_0^T \dd{t} \int_{B_\ell} \mathrm{d}q \  \varphi(q)  \left[(\nabla V_N * \varrho_{N,t}) (q)- (\nabla V * \varrho_{t}) (q) \right] \cdot \int \mathrm{d}p\ \nabla_{p} \phi(p)  m_{N,t}(q,p) \bigg|\\
	  & 
	 \leq T \sup_{t \in [0,T]} \norm{(\nabla V_N * \varrho_{N,t}) - (\nabla V * \varrho_{t})}_{L^r (B_\ell)} \norm{\varphi \int \mathrm{d}p\ \nabla_{p} \phi(p)  m_{N,t}(\cdot,p) }_{L^{r'} (B_\ell)}\\
	 &
	 \leq T \sup_{t \in [0,T]} \norm{(\nabla V_N * \varrho_{N,t}) - (\nabla V * \varrho_{t})}_{L^r(B_\ell)}   \norm{ \varphi }_{L^{r'}(B_\ell)} \norm{\nabla_p \phi}_{L^1(\R^3)}\\
	 &
	 \leq C_T  \norm{(\nabla V_N * \varrho_{N,t}) - (\nabla V * \varrho_{t})}_{L^\infty([0,T]; L^r(B_\ell))} ,
	\end{align*}
where we use the fact that $0 \leq m_{N,t} \leq 1$ almost everywhere. Taking the limit $N \to \infty$ on both sides, then we have
\[
	\lim_{N \to \infty} \mathcal{A}_1 = 0.
	\]

We focus now on $\mathcal{A}_2$. We observe that since $\norm{m_{N,t}}_{L^\infty}$ is uniformly bounded, it is implied that there is a subsequence still denoted by $(m_{N,t})_{N\in \N}$ such that $m_{N,t} \rightharpoonup^* m_{t}$ in $L^{\infty}((0,T);L^\infty(\R^{3} \times \R^3 ))$ as $N \to \infty$. Since $\varphi(q)\nabla_{p} \phi(p) \cdot  (\nabla V * \varrho_{t})(q) \in L^{1}((0,T);L^1(\R^{3} \times \R^3))$, we have $\lim_{N \to \infty} \mathcal{A}_2 = 0$. This completes the proof of Theorem \ref{theorem_main}. 

\qed

\section{Estimates of residuals}\label{sec_est_residual}
To estimate the residuals outlined in \eqref{bbgky_remainder_1}, we first present the following important facts, which are used frequently in the proof: the $\hbar$-weighted Fourier transformation is given as
\begin{equation}\label{hbar_fourier}
\int \mathrm{d}y\ G(y) F(y) =  \int \mathrm{d}y\ G(y) \frac{1}{(2\pi \hbar)^3}\iint \mathrm{d}p_2\mathrm{d}v\ \hat{F}(v) e^{\frac{\mathrm{i}}{\hbar}p_2\cdot(y-v)},
\end{equation}
for any given function $ F, G \in L^2(\R^3)$.

The results in \cite{Chen2021} for the localized number operator and oscillation estimates are

\begin{Lemma}[Lemma 2.4 of \cite{Chen2021}]\label{N_hbar} For $t \geq 0$, let $\Psi_{N,t} \in \mathcal{F}_a$ with $\norm{\Psi_{N,t}} = 1$ and $R_1>0$ be the radius of a ball such that the volume is $1$. Then, for all $1\leq k \leq N$, we have
\begin{equation}\label{N_hbar_eq}
	\begin{split}
	\dotsint (\mathrm{d}q\mathrm{d}x)^{\otimes k} \left(\prod_{n=1}^k \rchi_{|x_n-q_n|\leq \sqrt{\hbar}R_1}\right) \gamma^{(k)}_{N,t} (x_1,\dots, x_k; x_1,\dots,x_k)
	\leq \hbar^{-\frac{3}{2}k},
	\end{split}
	\end{equation}
where $\rchi$ is a characteristic function.
\end{Lemma}

\begin{Lemma}[Lemma 2.5 of \cite{Chen2021}]\label{estimate_oscillation}
For $g \in C^\infty_0 (\R^3)$ and
\begin{equation}\label{estimate_oscillation_omega}
	\Omega_{\hbar} := \{x \in \R^3;\ \max_{1\leq j \leq 3} |x_j|\leq \hbar^\alpha \},
	\end{equation}
it holds that for every $\alpha \in (0,1)$, $s \in \N$, and $x \in \R^3\backslash \Omega_{\hbar}$,
\begin{equation}\label{estimate_oscillation_0}
	\left|\int_{\R^3} \mathrm{d}p\  e^{\frac{\rm i}{\hbar}p\cdot x} g(p)\right| \leq c_1 \hbar^{(1-\alpha)s},
	\end{equation}
where the constant $c_1$ depends on the compact support and the $W^{s,\infty}$-norm of the test function $g$.
\end{Lemma}

The estimate for the residual term $\tilde{\mathcal{R}}$ given in \eqref{bbgky_remainder_1} is obtained exactly as shown in \cite{Chen2021}, i.e.,
\begin{Proposition}[Proposition 2.4 of \cite{Chen2021}]\label{est_tilde_R}
Suppose that $f\in H^1(\R^3)$, $\|f\|_{L^2}=1$ and has compact support; then, we have the following bound for $\tilde{\mathcal{R}}$ in \eqref{BBGKY_k1}; i.e., for an arbitrarily small $\delta >0$, there exists $s(\delta)>0$ such that the following estimate holds for any test function $\varphi, \phi\in C^\infty_0(\R^{3})$
\[
	\left|\iint \mathrm{d}q \mathrm{d}p\ \varphi(q)\phi(p)\nabla_{q}\cdot \tilde{\mathcal{R}}(q,p) \right| \leq c_2 \hbar^{\frac{1}{2}-\delta},
	\]
where the constant $c_2$ depends on $\|\nabla\varphi\|_{L^\infty}$ and $\|\phi\|_{W^{s,\infty}}$.
\end{Proposition}

For the residual term  $\mathcal{R}$, we insert the terms 
\[
	\pm \nabla V_N(q-q_2)  \gamma_{N,t}^{(2)}(u_1,u_2;w_1,w_2) ,
	\]
and write into a sum $\mathcal{R} = \mathcal{R}_1 + \mathcal{R}_2$, where
\begin{equation}\label{remainder_total}
	\begin{aligned}
	\mathcal{R}_1 := & (2\pi)^3  \iint \mathrm{d}w_1 \mathrm{d}u_1 \iint \mathrm{d}w_2\mathrm{d}u_2 \iint \mathrm{d}q_2 \mathrm{d}p_2    \left(  f^\hbar_{q,p}(w)  \overline{f^\hbar_{q,p}(u)} \right)^{\otimes 2} \\
	&
	\left[\int_0^1 \mathrm{d}s\ \nabla V_N\big( su_1+ (1-s)w_1-w_2 \big) - \nabla V_N(q-q_2) \right]  \gamma_{N,t}^{(2)}(u_1,u_2;w_1,w_2),
	\\
	\mathcal{R}_2 := 
	& (2\pi)^3  \iint \mathrm{d}w_1 \mathrm{d}u_1 \iint \mathrm{d}w_2\mathrm{d}u_2 \iint \mathrm{d}q_2 \mathrm{d}p_2    \left(  f^\hbar_{q,p}(w)  \overline{f^\hbar_{q,p}(u)} \right)^{\otimes 2} \nabla V_N(q-q_2) \\
	&\qquad  \left[\gamma_{N,t}^{(2)}(u_1,u_2;w_1,w_2)  - \gamma_{N,t}^{(1)}(u_1;w_1)\gamma_{N,t}^{(1)}(u_2;w_2)\right].
	\end{aligned}
	\end{equation}
Note that in \eqref{remainder_total}, $\mathcal{R}_1$ represents the semiclassical limit part, and $\mathcal{R}_2$ represents the mean-field limit.

As a preparation for the estimates of the residual term $\mathcal{R}$, we present the following estimate for regularized Coulomb potential:
\begin{Lemma}\label{lem_reg_V_N}
Let $V_N$ be the regularized Coulomb potential given in \eqref{def_V_N}, then it holds that
\begin{equation}\label{coulomb_00}
		\|\nabla V_N\|_{L^\infty}\leq C \beta_N^{-2}.
	\end{equation}
\end{Lemma}

\begin{proof}
Let $\widehat{V}_N$ be the Fourier transform of $V_N$. Recall that for $V_N =|\cdot|^{-1}*\mathcal{G}_{\beta_N}$, \cite[Theorem 5.9]{Lieb2001} yields
\[
	\widehat{V}_N(p) =  \frac{C}{|p|^2} e^{-\left(p\frac{\beta_N}{2}\right)^2 },
	\]
for some positive constant $C$. Then, by inverse Fourier transform, we get
\begin{align*}
	|\nabla V_N (x) | & = \frac{1}{(2\pi )^{3/2}} \left|\nabla_x  \int \dd{p} e^{ip\cdot x}   \widehat{V}_N(p)\right| \\
	&
	\leq C \int \dd{p} |p| |\widehat{V}_N(p)| \\
	&
	= C \int \dd{p}  \frac{1}{|p|} e^{-\left(p\frac{  \beta_N}{2}\right)^2}\\
	&
	\leq C \beta_N^{-2},
\end{align*}
where we use the spherical coordination in the last inequality.
\end{proof}

In the following, we treat the semiclassical and mean-field residual terms, i.e. $\mathcal{R}_1$ and $\mathcal{R}_2$, by using the truncated radius $\beta_N$, the oscillation estimate, the cutoff number operator and the kinetic operator estimates outlined in Section \ref{sec_proof_main}. 
\subsection{Estimate for the semiclassical residual term $\mathcal{R}_1$}
In this subsection, we present in full detail the estimate for the semiclassical residue.
\begin{Proposition}\label{est_semi_R}
Let $\varphi, \phi \in C_0^\infty(\R^3)$. Then, for  $\frac{5}{6} < \alpha_1 < 1$, $0 < \delta < \frac{1}{8}(6\alpha_1 - 5)$,  and $s = \left\lceil \frac{3(2\alpha_1 +1)}{4(1-\alpha_1)}\right\rceil$, we have
\begin{equation}
	\begin{aligned}
	\bigg|\iint \mathrm{d}q \mathrm{d}p\ \varphi(q)\phi(p)  \nabla_{p} \cdot \mathcal{R}_1(q,p) \bigg| & \leq  \tilde{C} \hbar^{\frac{1}{4}(6\alpha_1 - 5)-2\delta} ,
	\end{aligned}
	\end{equation}
where the constant  $\tilde{C}$ depends on $\norm{\varphi}_{W^{1,\infty}}$, $\norm{\nabla \phi}_{L^1 \cap W^{s,\infty}}$, $\supp \phi$, $\norm{f}_{L^\infty \cap H^1}$,  and $\supp f$.
\end{Proposition}

\begin{proof}
Recall from \eqref{remainder_total} that we have
\begin{equation}
	\begin{split}
	\mathcal{R}_1 := & (2\pi)^3  \iint \mathrm{d}w_1 \mathrm{d}u_1 \iint \mathrm{d}w_2\mathrm{d}u_2 \iint \mathrm{d}q_2 \mathrm{d}p_2    \left(  f^\hbar_{q,p}(w)  \overline{f^\hbar_{q,p}(u)} \right)^{\otimes 2} \\
	&\qquad \left[\int_0^1 \mathrm{d}s\ \nabla V_N\big(su_1+ (1-s)w_1 - w_2 \big) - \nabla V_N(q-q_2) \right]  \gamma_{N,t}^{(2)}(u_1,u_2;w_1,w_2).
	\end{split}
	\end{equation}
For $\varphi, \phi \in C_0^\infty(\R^3)$, we have
\begin{align*}
	&\bigg|\iint \mathrm{d}q \mathrm{d}p\ \varphi(q) \phi(p)  \nabla_{p} \cdot \mathcal{R}_1(q,p) \bigg| \\
	&
	=   (2\pi)^3 \bigg| \dotsint (\mathrm{d}q \mathrm{d}p)^{\otimes 2}\ \varphi(q) \nabla_{p} \phi(p) \cdot \iint \mathrm{d}w_1 \mathrm{d}u_1 \iint \mathrm{d}w_2\mathrm{d}u_2  \left(  f^\hbar_{q,p}(w)  \overline{f^\hbar_{q,p}(u)} \right)^{\otimes 2} \\
	&
	\qquad \left[\int_0^1 \mathrm{d}s\ \nabla V_N\big(su_1+ (1-s)w_1 - w_2 \big) - \nabla V_N(q-q_2) \right]  \gamma_{N,t}^{(2)}(u_1,u_2;w_1,w_2)\bigg|\\
	&
	= (2\pi)^6  \bigg| \iint (\mathrm{d}q)^{\otimes 2}\mathrm{d}p  \iiint \mathrm{d}w_1 \mathrm{d}u_1  \mathrm{d}w_2\ \varphi(q)  \nabla  \phi(p) \cdot f\left(\frac{w_1 -q}{\sqrt{\hbar}}\right) f\left(\frac{u_1 -q}{\sqrt{\hbar}}\right)\\
	&
	\qquad e^{\frac{\mathrm{i}}{\hbar}p \cdot (w_1-u_1)} \left|f\left(\frac{w_2 -q_2}{\sqrt{\hbar}}\right) \right|^2  \left[\int_0^1 \mathrm{d}s\ \nabla V_N\big(su_1+ (1-s)w_1 - w_2 \big) - \nabla V_N(q-q_2) \right] \\
	&
	\qquad \gamma_{N,t}^{(2)}(u_1,w_2;w_1,w_2)\bigg|,
	\end{align*}
where we apply the fact that $ (2\pi \hbar)^3 \delta_x(y) = \int  e^{\frac{\mathrm{i}}{\hbar} p \cdot (x-y)} \dd{p}$. Then, inserting $\pm \nabla V_N(q-w_2)$, by the triangle inequality, we have
\begin{align*}
	&
	\leqslant  (2\pi)^6  \bigg| \iiint (\mathrm{d}q)^{\otimes 2}\mathrm{d}p \iiint \mathrm{d}w_1 \mathrm{d}u_1  \mathrm{d}w_2\ \varphi(q)  \nabla  \phi(p) \cdot f\left(\frac{w_1 -q}{\sqrt{\hbar}}\right) f\left(\frac{u_1 -q}{\sqrt{\hbar}}\right)\\
	&
	\qquad e^{\frac{\mathrm{i}}{\hbar}p \cdot (w_1-u_1)} \left|f\left(\frac{w_2 -q_2}{\sqrt{\hbar}}\right) \right|^2  \left[\nabla_{w_2} \int_0^1 \mathrm{d}s\ V_N\big(su_1+ (1-s)w_1 - w_2 \big)  - V_N(q - w_2) \right] \\
	&
	\qquad \gamma_{N,t}^{(2)}(u_1,w_2;w_1,w_2)\bigg|\\
	& 
	\quad + (2\pi)^6  \bigg| \iiint (\mathrm{d}q)^{\otimes 2}\mathrm{d}p  \iiint \mathrm{d}w_1 \mathrm{d}u_1  \mathrm{d}w_2\ \varphi(q) \nabla \phi(p) \cdot f\left(\frac{w_1 -q}{\sqrt{\hbar}}\right) f\left(\frac{u_1 -q}{\sqrt{\hbar}}\right)\\
	&
	\qquad e^{\frac{\mathrm{i}}{\hbar}p \cdot (w_1-u_1)} \left|f\left(\frac{w_2 -q_2}{\sqrt{\hbar}}\right) \right|^2  \nabla_{q}  \left[ V_N(q - w_2 )  - V_N(q - q_2) \right] \\
	&
	\qquad \gamma_{N,t}^{(2)}(u_1,w_2;w_1,w_2)\bigg|\\
	& 
	= (2\pi)^6  \bigg| \iiint (\mathrm{d}q)^{\otimes 2}\mathrm{d}p  \iiint \mathrm{d}w_1 \mathrm{d}u_1  \mathrm{d}w_2\ \varphi(q)\nabla \phi(p) \cdot f\left(\frac{w_1 -q}{\sqrt{\hbar}}\right) f\left(\frac{u_1 -q}{\sqrt{\hbar}}\right)\\
	&
	\qquad e^{\frac{\mathrm{i}}{\hbar}p \cdot (w_1-u_1)} \left|f\left(\frac{w_2 -q_2}{\sqrt{\hbar}}\right) \right|^2  \left[  \int_0^1 \mathrm{d}s\ V_N\big(su_1+ (1-s)w_1 - w_2 \big)  - V_N(q - w_2) \right] \\
	&
	\qquad \nabla_{w_2} \gamma_{N,t}^{(2)}(u_1,w_2;w_1,w_2)\bigg|\\
	& 
	\quad + (2\pi)^6  \bigg| \iiint (\mathrm{d}q)^{\otimes 2}\mathrm{d}p  \iiint \mathrm{d}w_1 \mathrm{d}u_1  \mathrm{d}w_2\ \varphi(q)\nabla \phi(p) \cdot f\left(\frac{w_1 -q}{\sqrt{\hbar}}\right) f\left(\frac{u_1 -q}{\sqrt{\hbar}}\right)\\
	&
	\qquad e^{\frac{\mathrm{i}}{\hbar}p \cdot (w_1-u_1)} \nabla_{w_2} \left|f\left(\frac{w_2 -q_2}{\sqrt{\hbar}}\right) \right|^2  \left[ \int_0^1 \mathrm{d}s\ V_N\big(su_1+ (1-s)w_1 - w_2 \big)  - V_N(q - w_2) \right] \\
	&
	\qquad  \gamma_{N,t}^{(2)}(u_1,w_2;w_1,w_2)\bigg|\\
	& 
	\quad + (2\pi)^6  \bigg| \iiint (\mathrm{d}q)^{\otimes 2}\mathrm{d}p \iiint \mathrm{d}w_1 \mathrm{d}u_1  \mathrm{d}w_2\ \nabla \phi(p) \cdot \nabla_{q} \left(\varphi(q) f\left(\frac{w_1 -q}{\sqrt{\hbar}}\right) f\left(\frac{u_1 -q}{\sqrt{\hbar}}\right)\right)\\
	&
	\qquad e^{\frac{\mathrm{i}}{\hbar}p \cdot (w_1-u_1)} \left|f\left(\frac{w_2 -q_2}{\sqrt{\hbar}}\right) \right|^2  \left[    V_N\big(q - w_2 \big)  - V_N(q - q_2) \right] \\
	&
	\qquad \gamma_{N,t}^{(2)}(u_1,w_2;w_1,w_2)\bigg|,\\
	& =: I_1 + J_1 + K_1
	\end{align*}
where we use integration by parts in the second to last equality.\\

Before advancing, we observe that by splitting the integral with respect to momentum space $\Omega_{\hbar}$ and $\Omega_{\hbar}^c$ as defined in \eqref{estimate_oscillation_omega}, for constant $C_1$ depending on $\norm{\nabla  \phi}_{L^1 \cap W^{s,\infty}}$ and $\supp \phi$, we have
\begin{equation}\label{momentum_split}
	\begin{aligned}
	\left|\int \dd{p} \nabla  \phi(p) e^{\frac{\mathrm{i}}{\hbar}p\cdot(w-u)}\right|
	&=
	\left|\int \dd{p}  (\rchi_{(w_1-u_1)\in \Omega_\hbar}+ \rchi_{(w_1-u_1)\in \Omega^c_\hbar}) \phi(p) e^{\frac{\mathrm{i}}{\hbar}p\cdot(w-u)}\right|\\
	&
	\leq  C_1 \left(\rchi_{(w_1-u_1)\in \Omega_\hbar} + \hbar^{(1-\alpha_1)s}\right),
	\end{aligned}
	\end{equation}
where we use \eqref{estimate_oscillation_0} in the last inequality.

Now, we want to separately estimate the terms $I_1$ and $J_1$. We begin by estimating $I_1$. Recall that
\begin{align*}
	I_1 &
	= (2\pi)^6 \hbar^{\frac{3}{2}} \bigg| \iint \mathrm{d}q\mathrm{d}p \iiint \mathrm{d}w_1 \mathrm{d}u_1  \mathrm{d}w_2\ \varphi(q)\nabla  \phi(p) \cdot f\left(\frac{w_1 -q}{\sqrt{\hbar}}\right) f\left(\frac{u_1 -q}{\sqrt{\hbar}}\right)\\
	&
	\qquad e^{\frac{\mathrm{i}}{\hbar}p \cdot (w_1-u_1)}\left( \int \mathrm{d}\tilde{q}_2\left|f\left(\tilde{q}_2\right) \right|^2 \right) \left[ \int  \mathrm{d}s\ V_N\big(su_1+ (1-s)w_1 - w_2 \big)  - V_N(q - w_2) \right] \\
	&
	\qquad \nabla_{w_2} \gamma_{N,t}^{(2)}(u_1,w_2;w_1,w_2)\bigg|.\\
\intertext{By using \eqref{momentum_split} we have,}
	I_1 &
	{\leq}  \norm{\nabla V_N}_{L^\infty}C_1  \hbar^{\frac{3}{2}} \int \mathrm{d}q\ |\varphi(q )|\iint \mathrm{d}w_1 \mathrm{d}u_1  \mathrm{d}w_2 \left( \rchi_{(w_1-u_1)\in \Omega_\hbar} + \hbar^{(1-\alpha_1)s} \right)   \bigg|f\left(\frac{w_1 -q}{\sqrt{\hbar}}\right) f\left(\frac{u_1 -q}{\sqrt{\hbar}}\right)\bigg|   \\
	&
	\qquad \left( |u_1-q|+|w_1-q| \right) \big| \nabla_{w_2} \gamma_{N,t}^{(2)}(u_1,w_2;w_1,w_2)\big|\\
	&
	\leq \norm{\nabla V_N}_{L^\infty} C_1  \hbar^{\frac{3}{2}} \int \mathrm{d}q\ |\varphi(q )|\iint \mathrm{d}w_1 \mathrm{d}u_1   \left( \rchi_{(w_1-u_1)\in \Omega_\hbar} + \hbar^{(1-\alpha_1)s} \right)   \bigg|f\left(\frac{w_1 -q}{\sqrt{\hbar}}\right) f\left(\frac{u_1 -q}{\sqrt{\hbar}}\right)\bigg|   \\
	&
	\qquad \left( |u_1-q|+|w_1-q| \right) \int \mathrm{d}w_2 \big| \nabla_{w_2} \gamma_{N,t}^{(2)}( w_1,w_2; u_1, w_2)\big|\\
	&
	= \norm{\nabla V_N}_{L^\infty} C_1 \hbar^{\frac{3}{2}} \int \mathrm{d}q\ |\varphi(q )|\iint \mathrm{d}w_1 \mathrm{d}u_1  \left( \rchi_{(w_1-u_1)\in \Omega_\hbar} + \hbar^{(1-\alpha_1)s} \right)   \bigg|f\left(\frac{w_1 -q}{\sqrt{\hbar}}\right) f\left(\frac{u_1 -q}{\sqrt{\hbar}}\right)\bigg|   \\
	&
	\qquad \left( |u_1-q|+|w_1-q| \right) \int \mathrm{d}w_2 \big| \nabla_{w_2} \left<a_{w_2}a_{w_1} \Psi_{N,t},  a_{w_2}  a_{u_1} \Psi_{N,t} \right>\big|\\
	&
	\leq \norm{\nabla V_N}_{L^\infty}  C_1  \hbar^{\frac{3}{2}} \int \mathrm{d}q\ |\varphi(q )|\iint \mathrm{d}w_1 \mathrm{d}u_1   \left( \rchi_{(w_1-u_1)\in \Omega_\hbar} + \hbar^{(1-\alpha_1)s} \right)   \bigg|f\left(\frac{w_1 -q}{\sqrt{\hbar}}\right) f\left(\frac{u_1 -q}{\sqrt{\hbar}}\right)\bigg|   \\
	&
	\qquad \left( |u_1-q|+|w_1-q| \right) \int \mathrm{d}w_2 \bigg[ \norm{\nabla_{w_2} a_{w_2}a_{w_1} \Psi_{N,t}} \norm{  a_{w_2}  a_{u_1} \Psi_{N,t}} +  \norm{ a_{w_2}a_{w_1} \Psi_{N,t}} \norm{ \nabla_{w_2} a_{w_2}  a_{u_1} \Psi_{N,t}} \bigg]\\
	&
	\leq \norm{\nabla V_N}_{L^\infty} C_1  \hbar^{\frac{3}{2}} \int \mathrm{d}q\ |\varphi(q )|\iint \mathrm{d}w_1 \mathrm{d}u_1   \left( \rchi_{(w_1-u_1)\in \Omega_\hbar} + \hbar^{(1-\alpha_1)s} \right)   \bigg|f\left(\frac{w_1 -q}{\sqrt{\hbar}}\right) f\left(\frac{u_1 -q}{\sqrt{\hbar}}\right)\bigg|   \\
	&
	\qquad \left( |u_1-q|+|w_1-q| \right) \bigg[ \left(\int \mathrm{d}w_2 \norm{\nabla_{w_2} a_{w_2}a_{w_1} \Psi_{N,t}}^2\right)^\frac{1}{2} \left(\int \mathrm{d}w_2 \norm{  a_{w_2}  a_{u_1} \Psi_{N,t}^2}\right)^\frac{1}{2} \\
	&
	\qquad + \left(\int \mathrm{d}w_2 \norm{ a_{w_2}a_{w_1} \Psi_{N,t}}^2\right)^\frac{1}{2} \left(\int \mathrm{d}w_2\norm{ \nabla_{w_2} a_{w_2}  a_{u_1} \Psi_{N,t}}^2\right)^\frac{1}{2} \bigg]\\
	& =: C_1 \norm{\nabla V_N}_{L^\infty} \bigg[ i_{1,1} + i_{1,2}\bigg],
	\end{align*}

Before we continue, we observe that from the definition of the kinetic energy operator $\mathcal{K}$ and number operator $\mathcal{N}$, we have
\begin{equation}\label{est_mixture}
	\begin{split}
	\int &\mathrm{d}q\ |\varphi(q )| \left[\iint \mathrm{d}w_1 \mathrm{d}u_1\ \rchi_{|w_1 - q|\leq R_1 \sqrt{\hbar}} \rchi_{|u_1 - q|\leq R_1  \sqrt{\hbar}}  \left(\int \mathrm{d}w_2 \norm{\nabla_{w_2} a_{w_2} a_{w_1} \Psi_{N,t}}^2\right) \left(\int \mathrm{d}w_2 \norm{  a_{w_2}  a_{u_1} \Psi_{N,t}}^2\right)\right]^\frac{1}{2}\\
	&
	= 2\hbar^{-1} \int \mathrm{d}q\ |\varphi(q )| \left[\iint \mathrm{d}w_1 \mathrm{d}u_1 \ \rchi_{|w_1 - q|\leq R_1 \sqrt{\hbar}} \rchi_{|u_1 - q|\leq R_1  \sqrt{\hbar}}  \left<\Psi_{N,t}, a^*_{w_1} \mathcal{K} a_{w_1} \Psi_{N,t}\right> \left<\Psi_{N,t}, a^*_{u_1} \mathcal{N} a_{u_1} \Psi_{N,t}\right>\right]^\frac{1}{2}\\
	&
	=  2 \hbar^{-1} \int \mathrm{d}q\  |\varphi(q )| \bigg[ \int \mathrm{d}w_1 \ \rchi_{|w_1 - q|\leq R_1 \sqrt{\hbar}} \left<\Psi_{N,t}, \mathcal{K}  (a^*_{w_1}  a_{w_1}  -1) \Psi_{N,t}\right>  \\
	&
	\qquad \int \mathrm{d}u_1 \ \rchi_{|u_1 - q|\leq R_1 \sqrt{\hbar}} \left<\Psi_{N,t}, (\mathcal{N} -1 )a^*_{u_1}  a_{u_1} \Psi_{N,t}\right> \bigg]^\frac{1}{2}\\
	&
	\leq 2 \hbar^{-1}\left( \iint \mathrm{d}q \mathrm{d}w_1\  |\varphi(q )| \ \rchi_{|w_1 - q|\leq R_1 \sqrt{\hbar}} \left<\Psi_{N,t}, \mathcal{K} a^*_{w_1}  a_{w_1} \Psi_{N,t}\right>\right)^\frac{1}{2}  \\
	&
	\qquad \left(\iint \mathrm{d}q \mathrm{d}u_1\  |\varphi(q )| \rchi_{|u_1 - q|\leq R_1 \sqrt{\hbar}} \left<\Psi_{N,t}, \mathcal{N} a^*_{u_1}  a_{u_1} \Psi_{N,t}\right>\right)^\frac{1}{2}\\
	&
	\leq C_2  \hbar^{-4-\frac{3}{2}},
	\end{split}
	\end{equation}
where in the last step we use a direct outcome of \eqref{N_hbar_eq} and \eqref{kinetic_finite}, i.e.
\begin{equation}
	\begin{split}
	&\iint \mathrm{d}q \mathrm{d}x\ \rchi_{|x - q|\leq R_1 \sqrt{\hbar}} |\varphi(q)| \left<\Psi_{N,t},a^*_x \mathcal{K} a_{x}\Psi_{N,t} \right> \\
	& 
	= \iint \mathrm{d}q \mathrm{d}x\ \rchi_{|x - q|\leq R_1 \sqrt{\hbar}} |\varphi(q)| \left<\Psi_{N,t}, \mathcal{K} (a^*_x  a_{x} - 1 ) \Psi_{N,t} \right> \\
	&
	\leq \iint \mathrm{d}q \mathrm{d}x\ \rchi_{|x - q|\leq R_1 \sqrt{\hbar}} |\varphi(q)| \left<\Psi_{N,t}, \mathcal{K} a^*_x  a_{x}\Psi_{N,t} \right> \\
	&
	=  \left<\Psi_{N,t}, \mathcal{K} \iint \mathrm{d}q \mathrm{d}x\ \rchi_{|x - q|\leq R_1 \sqrt{\hbar}} |\varphi(q)| a^*_x  a_{x}\Psi_{N,t} \right> \\
	&
	\leq  C_2 \hbar^{-\frac{3}{2}}  \left<\Psi_{N,t}, \mathcal{K} \Psi_{N,t} \right>\\
	&
	\leq   C_2  \hbar^{-\frac{3}{2}-3}.
	\end{split}
	\end{equation}
In the above estimate, $C_2$ is a constant depends on $\norm{\varphi}_{L^\infty}$ and $\supp f$.

To continue, we apply the H\"older inequality to $i_{1,1}$ with respect to the terms $w_1$ and $u_1$,
\begin{align*}
	i_{1,1} &\leq  \hbar^{\frac{3}{2}} \int \mathrm{d}q\ |\varphi(q )| \bigg[\iint \mathrm{d}w_1 \mathrm{d}u_1    \rchi_{(w_1-u_1)\in \Omega_\hbar}   \bigg|f\left(\frac{w_1 -q}{\sqrt{\hbar}}\right) f\left(\frac{u_1 -q}{\sqrt{\hbar}}\right)\bigg|^2  \left(|u_1 - q|+|w_1-q| \right)^2 \bigg]^\frac{1}{2}\\
	&
	\qquad  \bigg[\iint \mathrm{d}w_1 \mathrm{d}u_1\  \rchi_{|w_1 - q|\leq R_1 \sqrt{\hbar}} \rchi_{|u_1 - q|\leq R_1  \sqrt{\hbar}} \left(\int \mathrm{d}w_2 \norm{\nabla_{w_2} a_{w_2} a_{w_1} \Psi_{N,t}}^2\right) \left(\int \mathrm{d}w_2 \norm{  a_{w_2}  a_{u_1} \Psi_{N,t}}^2 \right) \\
	&
	\qquad + \iint \mathrm{d}w_1 \mathrm{d}u_1\ \rchi_{|w_1 - q|\leq R_1 \sqrt{\hbar}} \rchi_{|u_1 - q|\leq R_1  \sqrt{\hbar}} \left(\int \mathrm{d}w_2 \norm{ a_{w_2}a_{w_1} \Psi_{N,t}}^2\right) \left(\int \mathrm{d}w_2\norm{ \nabla_{w_2} a_{w_2}  a_{u_1} \Psi_{N,t}}^2\right) \bigg]^\frac{1}{2}\\
	&
	= \hbar^{\frac{3}{2}}  \bigg[ \hbar^3 \iint \mathrm{d}\tilde{w}_1 \mathrm{d}\tilde{u}_1    \rchi_{|\tilde{w}_1-\tilde{u}_1|\leq  \hbar^{\alpha_1+\frac{1}{2}}}    |f\left(\tilde{w}_1\right) f\left(\tilde{u}\right)|^2 \hbar\left(|\tilde{u}_1|+|\tilde{w}_1| \right)^2 \bigg]^\frac{1}{2}\\
	&
	\qquad \int \mathrm{d}q\ |\varphi(q )|  \bigg[ 2 \iint \mathrm{d}w_1 \mathrm{d}u_1\  \rchi_{|w_1 - q|\leq R_1 \sqrt{\hbar}} \rchi_{|u_1 - q|\leq R_1  \sqrt{\hbar}} \\
	&
	\qquad \left(\int \mathrm{d}w_2 \norm{\nabla_{w_2} a_{w_2} a_{w_1} \Psi_{N,t}}^2\right) \left(\int \mathrm{d}w_2 \norm{  a_{w_2}  a_{u_1} \Psi_{N,t}}^2\right) \bigg]^\frac{1}{2}.\\
\intertext{Then by using \eqref{est_mixture}, the estimate goes further}
	&
	\leq C_2\hbar^{3-4-\frac{3}{2}+\frac{1}{2}} \bigg[ \iint \mathrm{d}\tilde{w}_1 \mathrm{d}\tilde{u}_1   \rchi_{|\tilde{w}_1-\tilde{u}_1|\leq R_1 \hbar^{\alpha_1+\frac{1}{2}}}   |f\left(\tilde{w}_1\right) f\left(\tilde{u}\right)|^2  \left(|\tilde{u}_1|+|\tilde{w}_1| \right)^2 \bigg]^\frac{1}{2}\\
	&
	\leq C_2\hbar^{-2} \bigg[ \iint \mathrm{d}\tilde{w}_1 \mathrm{d}\tilde{u}_1\    \rchi_{|\tilde{w}_1-\tilde{u}_1|\leq  R_1 \hbar^{\alpha_1+\frac{1}{2}}}    |f\left(\tilde{w}_1\right) f\left(\tilde{u}\right)|^2  \left| \tilde{u}_1+ \tilde{w}_1  \right|^2 \bigg]^\frac{1}{2}\\
	&
	\leq  C_3\hbar^{-2+\frac{3}{2}(\alpha_1 +\frac{1}{2})} =C_3 \hbar^{\frac{6\alpha_1 -5}{4}},
	\end{align*}
where $C_3$ depends on $\norm{\varphi}_{L^\infty}$, $\norm{f}_{L^\infty \cap L^2}$, $\supp f$ and we use the following estimate in the last inequality above:
\begin{equation}\label{bound_rchi_f}
	\begin{aligned}
	\int \mathrm{d}& w\ |f(w)|^2 \int \mathrm{d}u\ \rchi_{|w-u| \leq R_1 \hbar^{\alpha_1+\frac{1}{2}}}  |f(u)|^2\\
	&\leq \sup_{u} {|f(u)|^2}\int \mathrm{d}w\ |f(w)|^2 \int \mathrm{d}u\ \rchi_{|w-u| \leq R_1 \hbar^{\alpha_1+\frac{1}{2}}} \\
	&\leq  \norm{f}_{L^\infty}\norm{f}_{L^2} \hbar^{3\left(\alpha_1+\frac{1}{2}\right)},
	\end{aligned}
	\end{equation}
where the fixed radius $R_1$ arises from the compactness assumption of $f$.

With steps similar to those for $i_{1,1}$, we have
\begin{align*}
	i_{1,2} & \leq  C_3 \hbar^{-2}  \hbar^{(1-\alpha_1)s} \bigg[ \iint \mathrm{d}\tilde{w}_1 \mathrm{d}\tilde{u}_1     |f\left(\tilde{w}_1\right) f\left(\tilde{u}\right)|^2  \left(|\tilde{u}_1|+|\tilde{w}_1| \right)^2 \bigg]^\frac{1}{2} 
	\leq  C_4 \hbar^{-2+(1-\alpha_1)s},
	\end{align*}
where the constant $C_4$ depends on $\norm{\varphi}_\infty$, $\norm{f}_{L^\infty \cap L^2}$, and $\supp f$.
To balance the order between $i_{1,1}$ and $i_{1,2}$, $s$ is chosen to be
\[
	s = \left\lceil \frac{3(2\alpha_1 +1)}{4(1-\alpha_1)}\right\rceil,
	\]
for $\alpha_1 \in (\frac{5}{6},1)$. Therefore, we have
\begin{equation}\label{R1_I}
	I_1 \leq \tilde{C}  \norm{\nabla V_N}_{L^\infty}\hbar^{\frac{6\alpha_1 -5}{4}}.
	\end{equation}

To estimate $J_1$, we compute
\begin{align*}
	J_1 & 
	= (2\pi)^6  \bigg| \iint \mathrm{d}q \mathrm{d}p  \iiint \mathrm{d}w_1 \mathrm{d}u_1  \mathrm{d}w_2\ \varphi(q)\nabla  \phi(p) \cdot f\left(\frac{w_1 -q}{\sqrt{\hbar}}\right) f\left(\frac{u_1 -q}{\sqrt{\hbar}}\right)\\
	&
	\qquad e^{\frac{\mathrm{i}}{\hbar}p \cdot (w_1-u_1)}  \left[2 \hbar \int \mathrm{d}\tilde{q}_2\ f\left(\tilde{q}_2 \right)  \nabla f\left(\tilde{q}_2 \right)  \right] \left[ \int_0^1 \mathrm{d}s\ V_N\big(su_1+ (1-s)w_1 - w_2 \big)  - V_N(q - w_2) \right] \\
	&
	\qquad  \gamma_{N,t}^{(2)}(u_1,w_2;w_1,w_2)\bigg|\\
	&
	\leq (2\pi)^6 \hbar \int \mathrm{d}q  |\varphi(q )|\iiint \mathrm{d}w_1 \mathrm{d}u_1  \mathrm{d}w_2 \left|\int \mathrm{d}p (\rchi_{(w_1-u_1)\in \Omega_\hbar}+ \rchi_{(w_1-u_1)\in \Omega^c_\hbar}) \nabla \phi(p) \cdot e^{\frac{\mathrm{i}}{\hbar}p \cdot (w_1-u_1)}\right|\\
	&
	\qquad  \left[ \int \mathrm{d}\tilde{q}_2\ |\nabla f\left(\tilde{q}_2 \right)|  \right]  \bigg|f\left(\frac{w_1 -q}{\sqrt{\hbar}}\right) f\left(\frac{u_1 -q}{\sqrt{\hbar}}\right)    \left[ \int_0^1 \mathrm{d}s\ V_N\big(su_1+ (1-s)w_1 - w_2 \big)  - V_N(q - w_2) \right]\\
	&
	\qquad  \gamma_{N,t}^{(2)}(u_1,w_2;w_1,w_2)\bigg|\\
	&
	\leq C_1 \hbar \int \mathrm{d}q  |\varphi(q )|\iiint \mathrm{d}w_1 \mathrm{d}u_1  \mathrm{d}w_2  \left( \rchi_{(w_1-u_1)\in \Omega_\hbar} + \hbar^{(1-\alpha_1)s} \right) \bigg|f\left(\frac{w_1 -q}{\sqrt{\hbar}}\right) f\left(\frac{u_1 -q}{\sqrt{\hbar}}\right)  \bigg| \\
	&
	\qquad \left[ \int_0^1 \mathrm{d}s\ |V_N\big(su_1+ (1-s)w_1 - w_2 \big)  - V_N(q - w_2)| \right]
	\bigg| \gamma_{N,t}^{(2)}(u_1,w_2;w_1,w_2)\bigg|\\
	&
	\leq C_1 \norm{\nabla V_N}_{L^\infty} \hbar \int \mathrm{d}q  |\varphi(q )|\iint \mathrm{d}w_1 \mathrm{d}u_1   \left( \rchi_{(w_1-u_1)\in \Omega_\hbar} + \hbar^{(1-\alpha_1)s} \right) \bigg|f\left(\frac{w_1 -q}{\sqrt{\hbar}}\right) f\left(\frac{u_1 -q}{\sqrt{\hbar}}\right)  \bigg| \\
	&
	\qquad\left( |u_1-q|+|w_1-q| \right) \int \mathrm{d}w_2\  \norm{a_{w_2} a_{w_1} \Psi_{N,t}}_2 \norm{a_{w_2} a_{u_1} \Psi_{N,t}}_2 \\
	&
	\leq  C_1 \norm{\nabla V_N}_{L^\infty}  \hbar \int \mathrm{d}q  |\varphi(q )|\iint \mathrm{d}w_1 \mathrm{d}u_1   \left( \rchi_{(w_1-u_1)\in \Omega_\hbar} + \hbar^{(1-\alpha_1)s} \right) \bigg|f\left(\frac{w_1 -q}{\sqrt{\hbar}}\right) f\left(\frac{u_1 -q}{\sqrt{\hbar}}\right)  \bigg| \\
	&
	\qquad\left( |u_1-q|+|w_1-q| \right) \left( \int \mathrm{d}w_2\  \norm{a_{w_2} a_{w_1} \Psi_{N,t}}_2^2 \right)^\frac{1}{2} \left( \int \mathrm{d}w_2\  \norm{a_{w_2} a_{u_1} \Psi_{N,t}}_2^2 \right)^\frac{1}{2} \\
	&
	=: C_1 \norm{\nabla V_N}_{L^\infty} \big[j_{1,1}+ j_{1,2}\big]. \numberthis \label{eq:j11j12}
	\end{align*}

As in part $I_1$, we separately analyze $j_{1,1}$ and $j_{1,2}$.
\begin{align*}
	j_{1,1}&
	=  \hbar \int \mathrm{d}q  |\varphi(q )|\iint \mathrm{d}w_1 \mathrm{d}u_1   \rchi_{(w_1-u_1)\in \Omega_\hbar}  \bigg|f\left(\frac{w_1 -q}{\sqrt{\hbar}}\right) f\left(\frac{u_1 -q}{\sqrt{\hbar}}\right)  \bigg| \left( |u_1-q|+|w_1-q| \right) \\
	&
	\quad \rchi_{|u_1 - q|\leq R_1 \sqrt{\hbar}}\rchi_{|w_1 - q|\leq R_1 \sqrt{\hbar}} \left( \int \mathrm{d}w_2\  \left<\Psi_{N,t}, a_{w_1}^* a_{w_2}^* a_{w_2} a_{w_1} \Psi_{N,t}\right>\right)^\frac{1}{2} \left( \int \mathrm{d}w_2\  \left<\Psi_{N,t}, a_{u_1}^* a_{w_2}^* a_{w_2} a_{u_1} \Psi_{N,t}\right> \right)^\frac{1}{2}\\
	&
	\leq \hbar \int \mathrm{d}q  |\varphi(q )| \left(\iint \mathrm{d}w_1 \mathrm{d}u_1   \rchi_{(w_1-u_1)\in \Omega_\hbar}  \bigg|f\left(\frac{w_1 -q}{\sqrt{\hbar}}\right) f\left(\frac{u_1 -q}{\sqrt{\hbar}}\right)  \bigg|^2 \left( |u_1-q|+|w_1-q| \right)^2\right)^\frac{1}{2}\\
	&
	\quad  \left(\iint \mathrm{d}w_1 \mathrm{d}w_2\ \rchi_{|w_1 - q|\leq R_1 \sqrt{\hbar}}   \left<\Psi_{N,t}, a_{w_1}^* a_{w_2}^* a_{w_2} a_{w_1} \Psi_{N,t}\right> \right)\\
	&
	= \hbar  \left(\hbar^3 \iint \mathrm{d}\tilde{w}_1 \mathrm{d}\tilde{u}_1   \rchi_{(\tilde{w}_1-\tilde{u}_1)\in \Omega_\hbar}  |f\left(\tilde{w}\right) f\left(\tilde{u}\right)  |^2 \hbar \left( |\tilde{u}|+|\tilde{w}| \right)^2\right)^\frac{1}{2} \iint \mathrm{d}q \mathrm{d}w_1\  \rchi_{|w_1 - q|\leq R_1 \sqrt{\hbar}}  |\varphi(q )|  \\
	&
	\qquad \left<\Psi_{N,t}, a_{w_1}^* \mathcal{N} a_{w_1} \Psi_{N,t}\right> \\
	&
	\leq \norm{\varphi}_{L^\infty} \hbar^{1+2-3-\frac{3}{2}} \left( \iint \mathrm{d}\tilde{w}_1 \mathrm{d}\tilde{u}_1   \rchi_{(\tilde{w}_1-\tilde{u}_1)\in \Omega_\hbar}  |f\left(\tilde{w}\right) f\left(\tilde{u}\right)  |^2  \left( |\tilde{u}|+|\tilde{w}| \right)^2\right)^\frac{1}{2} \\
	&
	{\leq} C_4 \hbar^{\frac{3}{2}(\alpha_1 - \frac{1}{2})},
	\end{align*}
where we use \eqref{bound_rchi_f} in the last inequality.

On the other hand, from the definition of $j_{1,2}$ in \eqref{eq:j11j12}, we get
\begin{align*}
	j_{1,2}&
	=  \hbar \int \mathrm{d}q  |\varphi(q )|\iint \mathrm{d}w_1 \mathrm{d}u_1\   \hbar^{(1-\alpha_1)s} \bigg|f\left(\frac{w_1 -q}{\sqrt{\hbar}}\right) f\left(\frac{u_1 -q}{\sqrt{\hbar}}\right)  \bigg|
	\left( |u_1-q|+|w_1-q| \right) \\
	&
	\qquad \rchi_{|u_1 - q|\leq R_1 \sqrt{\hbar}}\rchi_{|w_1 - q|\leq R_1 \sqrt{\hbar}} \left( \int \mathrm{d}w_2\  \norm{a_{w_2} a_{w_1} \Psi_{N,t}}_2^2 \right)^\frac{1}{2} \left( \int \mathrm{d}w_2\  \norm{a_{w_2} a_{u_1} \Psi_{N,t}}_2^2 \right)^\frac{1}{2}\\
	&
	\leq   \hbar^{1+(1-\alpha_1)s}\left(\iint \mathrm{d}w_1 \mathrm{d}u_1\   \bigg|f\left(\frac{w_1 -q}{\sqrt{\hbar}}\right) f\left(\frac{u_1 -q}{\sqrt{\hbar}}\right)  \bigg|^2
	\left( |u_1-q|+|w_1-q| \right)^2 \right)^\frac{1}{2}\\
	&
	\qquad \int \mathrm{d}q  |\varphi(q )| \iint  \mathrm{d}w_1 \mathrm{d}w_2\ \rchi_{|w_1 - q|\leq R_1 \sqrt{\hbar}}  \left<\Psi_{N,t}, a_{w_1}^* a_{w_2}^* a_{w_2} a_{w_1} \Psi_{N,t}\right>\\
	& 
	\leq   \hbar^{1+(1-\alpha_1)s-3 - \frac{3}{2}}  \bigg( \hbar^4 \iint \mathrm{d} \tilde{w}_1 \mathrm{d} \tilde{u}_1\ \rchi_{|\tilde{w}_1-\tilde{u}_1|\leq  \hbar^{\alpha_1+\frac{1}{2}}}   |f(\tilde{w}_1) f(\tilde{u})|^2  (|\tilde{u}_1|+|\tilde{w}_1| )^2 \bigg)^\frac{1}{2}\\
	&
	\leq C_4 \hbar^{(1-\alpha_1)s - \frac{3}{2}}.
	\end{align*}
To obtain the same order for $j_{1,1}$ and $j_{1,2}$, we can choose
\[
	s = \left\lceil \frac{3(2\alpha_1+1)}{4(1-\alpha_1)} \right\rceil.
	\]

Thus, for $\alpha_1 \in (\frac{1}{2},1)$, we have
\begin{equation}
	J_1 \leq  \tilde{C} \norm{\nabla V_N}_{L^\infty} \hbar^{\frac{3}{2}(\alpha_1 - \frac{1}{2})}.
	\end{equation}

Now, we want to estimate $K_1$.
\begin{align*}
	K_1 &= (2\pi)^6  \bigg| \iiint (\mathrm{d}q)^{\otimes 2}\mathrm{d}p \iiint \mathrm{d}w_1 \mathrm{d}u_1  \mathrm{d}w_2\  \nabla  \phi(p) \cdot \nabla_{q} \left(\varphi(q)  f\left(\frac{w_1 -q}{\sqrt{\hbar}}\right) f\left(\frac{u_1 -q}{\sqrt{\hbar}}\right)\right)\\
	&
	\qquad e^{\frac{\mathrm{i}}{\hbar}p \cdot (w_1-u_1)} \left|f\left(\frac{w_2 -q_2}{\sqrt{\hbar}}\right) \right|^2  \left[   V_N\big(q - w_2 \big)  - V_N(q - q_2) \right] \\
	&
	\qquad \gamma_{N,t}^{(2)}(u_1,w_2;w_1,w_2)\bigg|\\
	&
	= (2\pi)^6  \bigg| \iiint (\mathrm{d}q)^{\otimes 2}\mathrm{d}p \iiint \mathrm{d}w_1 \mathrm{d}u_1  \mathrm{d}w_2\ \nabla  \phi(p) \cdot  \bigg[ \nabla \varphi(q)  f\left(\frac{w_1 -q}{\sqrt{\hbar}} \right)f\left(\frac{u_1 -q}{\sqrt{\hbar}}\right)\\
	&
	\qquad - \hbar^{-\frac{1}{2}} \varphi(q)\nabla  f\left(\frac{w_1 -q}{\sqrt{\hbar}} \right)  f\left(\frac{u_1 -q}{\sqrt{\hbar}}\right) - \hbar^{-\frac{1}{2}} \varphi(q)   f\left(\frac{w_1 -q}{\sqrt{\hbar}} \right)   \nabla f\left(\frac{u_1 -q}{\sqrt{\hbar}}\right) \bigg]\\
	&
	\qquad e^{\frac{\mathrm{i}}{\hbar}p \cdot (w_1-u_1)} \left|f\left(\frac{w_2 -q_2}{\sqrt{\hbar}}\right) \right|^2  \left[  V_N\big(q - w_2 \big)  - V_N(q - q_2) \right]  \gamma_{N,t}^{(2)}(u_1,w_2;w_1,w_2)\bigg|\\
	&
	= :k_{1,1} + k_{1,2} + k_{1,3}.
	\end{align*}

Note that, for any $\varphi \in C^\infty_0$ and $f \in W^{1,2}_0$, the term $k_{1,1}$ is $\sqrt{\hbar}$-order higher than $k_{1,2}$ and $k_{1,3}$. Moreover, the estimate of the terms $k_{1,2}$ and $k_{1,3}$ are the same when doing change of variables in the final steps. Therefore, we focus only on the term $k_{1,2}$.
\begin{align*}
	k_{1,2} &
	\leq (2\pi)^6 \hbar^{-\frac{1}{2}}  \iint (\mathrm{d}q)^{\otimes 2} \iiint \mathrm{d}w_1 \mathrm{d}u_1  \mathrm{d}w_2 \left|\int \mathrm{d}p\  \nabla  \phi(p) e^{\frac{\mathrm{i}}{\hbar}p \cdot (w_1-u_1)} \right|  \bigg| \varphi(q)  \nabla  f\left(\frac{w_1 -q}{\sqrt{\hbar}} \right)  f\left(\frac{u_1 -q}{\sqrt{\hbar}}\right)  \bigg|\\
	&
	\qquad  \left|f\left(\frac{w_2 -q_2}{\sqrt{\hbar}}\right) \right|^2  \left|  V_N\big(q - w_2 \big)  - V_N(q - q_2) \right|  |\gamma_{N,t}^{(2)}(u_1,w_2;w_1,w_2)|\\
	&
	= (2\pi)^6  \hbar^{-\frac{1}{2}}  \int \mathrm{d}q \iiint \mathrm{d}w_1 \mathrm{d}u_1  \mathrm{d}w_2 \left|\int \mathrm{d}p\ \left(\rchi_{(w_1-u_1) \in \Omega_{\hbar}} + \rchi_{(w_1-u_1) \in \Omega_{\hbar}^c}\right) \nabla \phi(p) \cdot e^{\frac{\mathrm{i}}{\hbar}p \cdot (w_1-u_1)} \right| \\
	&
	\qquad  \bigg| \varphi(q)  \nabla  f\left(\frac{w_1 -q}{\sqrt{\hbar}} \right)  f\left(\frac{u_1 -q}{\sqrt{\hbar}}\right)  \bigg|  \int \hbar^{\frac{3}{2}} \mathrm{d}\tilde{q}_2 \left|f(\tilde{q}_2) \right|^2  \left|   V_N\big(q - w_2 \big)  - V_N(q - \sqrt{\hbar}\tilde{q}_2 - w_2) \right|\\
	&
	\qquad  | \gamma_{N,t}^{(2)}(u_1,w_2;w_1,w_2)|\\
	&
	\leq C_1 \norm{\nabla V_N}_{L^\infty}  \hbar^{1+\frac{1}{2}} \int \mathrm{d}q \iint \mathrm{d}w_1 \mathrm{d}u_1   \left|\rchi_{(w_1-u_1) \in \Omega_{\hbar}} + \hbar^{(1-\alpha_1)s} \right| \\
	&
	\qquad  \bigg| \varphi(q)  \nabla  f\left(\frac{w_1 -q}{\sqrt{\hbar}} \right)  f\left(\frac{u_1 -q}{\sqrt{\hbar}}\right)  \bigg|  \left(\int  \mathrm{d}\tilde{q}_2 |\tilde{q}_2| \left|f(\tilde{q}_2) \right|^2  \right) \int \mathrm{d}w_2\ | \gamma_{N,t}^{(2)}(u_1,w_2;w_1,w_2)|\\
	&
	\leq C_1 \norm{\nabla V_N}_{L^\infty}  \hbar^{1+\frac{1}{2}} \int \mathrm{d}q\ |\varphi(q)| \iint \mathrm{d}w_1 \mathrm{d}u_1   \left(\rchi_{(w_1-u_1)\leq \hbar^{\alpha_1}} + \hbar^{(1-\alpha_1)s} \right)  \bigg| \nabla  f\left(\frac{w_1 -q}{\sqrt{\hbar}} \right)  f\left(\frac{u_1 -q}{\sqrt{\hbar}}\right)  \bigg|\\
	&
	\qquad \rchi_{|w_1-q|\leq R_1 \sqrt{\hbar}}\rchi_{|u_1-q|\leq R_1 \sqrt{\hbar}}
	\left(\int \mathrm{d}w_2\ \norm{a_{w_2} a_{w_1}\Psi_{N,t}}^2\right)^\frac{1}{2} \left(\int \mathrm{d}w_2\ \norm{a_{w_2} a_{u_1}\Psi_{N,t}}^2\right)^\frac{1}{2}\\
	&
	=: C_1 \norm{\nabla V_N}_{L^\infty} [ \tilde{k}_1 + \tilde{k}_2,].
	\end{align*}

Using the H\"older inequality with respect to $w_1$ and $u_1$, we obtain that
\begin{align*}
	\tilde{k}_1 &
	\leq  \hbar^{1+\frac{1}{2}} \int \mathrm{d}q\ |\varphi(q)| \left[\iint \mathrm{d}w_1 \mathrm{d}u_1\  \rchi_{(w_1-u_1)\leq \hbar^{\alpha_1}}   \bigg| \nabla  f\left(\frac{w_1 -q}{\sqrt{\hbar}} \right)  f\left(\frac{u_1 -q}{\sqrt{\hbar}}\right)  \bigg|^2 \right]^\frac{1}{2}\\
	&
	\qquad  \left[\iint \mathrm{d}w_1 \mathrm{d}u_1\   \rchi_{|w_1-q|\leq R_1 \sqrt{\hbar}}\rchi_{|u_1-q|\leq R_1 \sqrt{\hbar}}
	\left(\int \mathrm{d}w_2\ \norm{a_{w_2} a_{w_1}\Psi_{N,t}}^2\right) \left(\int \mathrm{d}w_2\ \norm{a_{w_2} a_{u_1}\Psi_{N,t}}^2\right)\right]^\frac{1}{2}\\
	&
	= \hbar^{\frac{3}{2}} \left[ \hbar^{3} \iint \mathrm{d}\tilde{w}_1 \mathrm{d}\tilde{u}_1\  \rchi_{(\tilde{w}_1-\tilde{u}_1)\leq \hbar^{\alpha_1}+\frac{1}{2}}   \bigg| \nabla  f\left(\tilde{w}_1 \right)  f\left(\tilde{u}_1\right)  \bigg|^2 \right]^\frac{1}{2}\\
	&
	\qquad \int \mathrm{d}q\ |\varphi(q)|  \left[ \int \mathrm{d}w_1\   \rchi_{|w_1-q|\leq R_1 \sqrt{\hbar}}
	\left<\Psi_{N,t}, a_{w_1}^*\mathcal{N} a_{w_1} \Psi_{N,t} \right> \right]\\
	&
	\leq C_5 \hbar^{\frac{3}{2}(\alpha_1+\frac{1}{2})-\frac{3}{2}}  = C_5  \hbar^{\frac{3(2\alpha_1-1)}{4}},
	\end{align*}
where we use \eqref{N_hbar_eq} in the last inequality and $C_5$ depending on $\norm{f}_{L^\infty}$, $\norm{\nabla f}_{L^2}$, $\supp f$, and $\norm{\varphi}_{L^\infty}$.

Similarly, to calculate $j_{1,2}$,
\begin{align*}
	\tilde{k}_2 & \leq \hbar^{1+\frac{1}{2}} \int \mathrm{d}q\ |\varphi(q)| \iint \mathrm{d}w_1 \mathrm{d}u_1  \hbar^{(1-\alpha_1)s}  \bigg| \nabla  f\left(\frac{w_1 -q}{\sqrt{\hbar}} \right)  f\left(\frac{u_1 -q}{\sqrt{\hbar}}\right)  \bigg|\\
	& \qquad \rchi_{|w_1-q|\leq R_1 \sqrt{\hbar}}\rchi_{|u_1-q|\leq R_1 \sqrt{\hbar}}
	\left(\int \mathrm{d}w_2\ \norm{a_{w_2} a_{w_1}\Psi_{N,t}}^2\right)^\frac{1}{2} \left(\int \mathrm{d}w_2\ \norm{a_{w_2} a_{u_1}\Psi_{N,t}}^2\right)^\frac{1}{2}\\
	&
	\leq C_5 \hbar^{ (1-\alpha_1)s - \frac{3}{2}},
	\end{align*}
where $s$ is chosen as
\[
	s = \left\lceil \frac{3(2\alpha_1 +1)}{4(1-\alpha_1)}\right\rceil,
	\]
for $\alpha_1 \in (\frac{1}{2},1)$. Thus,
\begin{equation}
	K_1  \leq \tilde{C} \norm{\nabla V_N}_{L^\infty}  \hbar^{\frac{3(2\alpha_1-1)}{4}},
	\end{equation}
where we recall that the constant $\tilde{C}$ depends on $\norm{\varphi}_{ W^{1,\infty}}$, $\norm{\nabla \phi}_{L^1 \cap W^{s,\infty}}$, $\supp\phi$, $\norm{f}_{L^\infty \cap H^1}$, and $\supp f$.

Therefore, in summary, we have
\begin{equation*}
		\bigg|\iint \mathrm{d}q \mathrm{d}p\ \varphi(q) \phi(p)  \nabla_{p} \cdot \mathcal{R}_1(q,p) \bigg|  \leq \tilde{C}  \norm{\nabla V_N}_{L^\infty} \hbar^{\frac{1}{4}(6\alpha_1 - 5)} \leq \tilde{C}  \beta_{N}^{-2}\hbar^{\frac{1}{4}(6\alpha_1 - 5)}  ,
	\end{equation*}
where we use \eqref{coulomb_00} in the second inequality.

Setting $\beta_{N} = \hbar^\delta$ for $0 < \delta < \frac{1}{8}(6\alpha_1 - 5)$, we obtain the desired result.
\end{proof}

\subsection{Estimate for the mean-field residual term $\mathcal{R}_2$}

\begin{Proposition}\label{est_meanfield_R}
Let $\varphi, \phi \in C_0^\infty(\R^3)$. Then, for $\frac{1}{2} < \alpha_2 < 1$, $0 < \delta < \frac{3}{4}(\alpha_2 - \frac{1}{2})$,  and $s = \left\lceil \frac{3(2\alpha_2 +1)}{4(1-\alpha_2)}\right\rceil$, we have
\begin{equation}
	\bigg|\iint \mathrm{d}q \mathrm{d}p\ \varphi(q)\phi(p)  \nabla_{p} \cdot \mathcal{R}_2(q,p) \bigg| 
	\leq \tilde{C} \hbar^{\frac{3}{2}(\alpha_2 - \frac{1}{2}) -2\delta}
	\end{equation}
where the constant $\tilde{C}$ depends on $\norm{\varphi}_\infty$, $\norm{\nabla \phi}_{L^1\cap W^{s,\infty}}$, $\norm{f}_{L^\infty \cap H^1}$, $\supp f$, and $\supp \phi$.
\end{Proposition}

\begin{proof}

Recall that from \eqref{remainder_total}, we have
\begin{equation}\label{splitting_R2}
	\begin{split}
	\mathcal{R}_2 := 
	& (2\pi)^3  \iint \mathrm{d}w_1 \mathrm{d}u_1 \iint \mathrm{d}w_2\mathrm{d}u_2 \iint \mathrm{d}q_2 \mathrm{d}p_2    \left(  f^\hbar_{q,p}(w)  \overline{f^\hbar_{q,p}(u)} \right)^{\otimes 2} \nabla V_N(q-q_2) \\
	&\qquad  \left[\gamma_{N,t}^{(2)}(u_1,u_2;w_1,w_2)  - \gamma_{N,t}^{(1)}(u_1;w_1)\gamma_{N,t}^{(1)}(u_2;w_2)\right].
	\end{split}
	\end{equation}

Then, we have
\begin{align*}
	&\bigg|\iint \mathrm{d}q \mathrm{d}p\ \varphi(q) \phi(p) \nabla_{p} \cdot \mathcal{R}_{2} \bigg| \\
	&
	= \bigg| \dotsint (\mathrm{d}q \mathrm{d}p)^{\otimes 2} (\mathrm{d}w \mathrm{d}u)^{\otimes 2}\ \varphi(q) \nabla\phi(p)  \cdot  \left(  f^\hbar_{q,p}(w)  \overline{f^\hbar_{q,p}(u)} \right)^{\otimes 2} \nabla V_N(q-q_2) \\
	&
	\qquad \left[\gamma_{N,t}^{(2)}(u_1,u_2;w_1,w_2)  - \gamma^{(1)}_{N,t}(u_1; w_1) \gamma^{(1)}_{N,t}(u_2; w_2)\right] \bigg|\\
	&
	= \hbar^{-3}\bigg| \dotsint (\mathrm{d}q \mathrm{d}p)^{\otimes 2} (\mathrm{d}w \mathrm{d}u)^{\otimes 2}\  \varphi(q) \nabla\phi(p)  \cdot \left(f\left(\frac{w-q}{\sqrt{\hbar}}\right) f\left(\frac{u-q}{\sqrt{\hbar}}\right) e^{\frac{\mathrm{i}}{\hbar} p\cdot(w-u)}\right)^{\otimes 2}  \nabla V_N(q-q_2) \\
	&
	\qquad \left[\gamma_{N,t}^{(2)}(u_1,u_2;w_1,w_2)  - \gamma^{(1)}_{N,t}(u_1; w_1) \gamma^{(1)}_{N,t}(u_2; w_2)\right] \bigg|\\
	&
	= \hbar^{-3}\bigg| \dotsint (\mathrm{d}q \mathrm{d}w \mathrm{d}u)^{\otimes 2}\   \left(f\left(\frac{w-q}{\sqrt{\hbar}}\right) f\left(\frac{u-q}{\sqrt{\hbar}}\right) \right)^{\otimes 2}\left(\int \mathrm{d}p\ \varphi(q) \nabla\phi(p)  \cdot e^{\frac{\mathrm{i}}{\hbar} p\cdot(w_1-u_1)}\right) \\
	&
	\qquad \nabla V_N(q-q_2) \left(\int \mathrm{d}p_2\ e^{\frac{\mathrm{i}}{\hbar} p_2\cdot(w_2-u_2)} \right) \left[\gamma_{N,t}^{(2)}(u_1,u_2;w_1,w_2)  -  \gamma^{(1)}_{N,t}(u_1; w_1) \gamma^{(1)}_{N,t}(u_2; w_2) \right] \bigg|\\
	&
	= (2 \pi )^3 \bigg| \dotsint (\mathrm{d}q)^{\otimes 2} \mathrm{d}w_1\mathrm{d}u_1\mathrm{d}w_2 \   f\left(\frac{w_1-q}{\sqrt{\hbar}}\right) f\left(\frac{u_1-q}{\sqrt{\hbar}}\right) \left|f\left(\frac{w_2-q_2}{\sqrt{\hbar}}\right)\right|^2  \left(\int \mathrm{d}p\ \varphi(q) \nabla\phi(p)  \cdot e^{\frac{\mathrm{i}}{\hbar} p\cdot(w_1-u_1)}\right) \\
	&
	\qquad \nabla V_N(q-q_2) \left[\gamma_{N,t}^{(2)}(u_1,w_2;w_1,w_2)  -  \gamma^{(1)}_{N,t}(u_1; w_1) \gamma^{(1)}_{N,t}(w_2; w_2)\right] \bigg|,
	\end{align*}
where we use the weighted Dirac-delta function in the last equality; i.e.,
\begin{equation*}\label{diract_delta}
	\frac{1}{(2\pi \hbar)^3}\int \mathrm{d}p_2\ e^{\frac{\mathrm{i}}{\hbar} p_2\cdot(w_2-u_2)} = \delta_{w_2}(u_2).
	\end{equation*}
Now, splitting the domains of $w_1$ and $u_1$ into two, namely, with the characteristic functions $\rchi_{(w_1 - u_1)\in \Omega_\hbar}$ and $\rchi_{(w_1 - u_1)\in\Omega_\hbar^c}$ as defined in \eqref{estimate_oscillation_omega}, we have

\begin{align*}
	&\leqslant (2 \pi )^3 \bigg| \iint (\mathrm{d}q)^{\otimes 2}  \varphi(q) \iiint \mathrm{d}w_1\mathrm{d}u_1\mathrm{d}w_2\   f\left(\frac{w_1-q}{\sqrt{\hbar}}\right) f\left(\frac{u_1-q}{\sqrt{\hbar}}\right) \left|f\left(\frac{w_2-q_2}{\sqrt{\hbar}}\right)\right|^2  \\
	&
	\qquad \left(\int \mathrm{d}p\ \rchi_{(w_1 - u_1)\in \Omega_\hbar} e^{\frac{\mathrm{i}}{\hbar} p\cdot(w_1-u_1)} \nabla\phi(p)   \right) \cdot \nabla V_N(q-q_2) \left[\gamma_{N,t}^{(2)}(u_1,w_2;w_1,w_2)  - \gamma^{(1)}_{N,t}(u_1; w_1) \gamma^{(1)}_{N,t}(w_2; w_2)\right] \bigg|\\
	&
	+ (2 \pi )^3 \bigg| 	\iint (\mathrm{d}q)^{\otimes 2}\varphi(q) \iiint \mathrm{d}w_1\mathrm{d}u_1\mathrm{d}w_2\    f\left(\frac{w_1-q}{\sqrt{\hbar}}\right) f\left(\frac{u_1-q}{\sqrt{\hbar}}\right) \left|f\left(\frac{w_2-q_2}{\sqrt{\hbar}}\right)\right|^2  \\
	&
	\qquad \left(\int \mathrm{d}p\ \rchi_{(w_1 - u_1)\in \Omega_\hbar^c}  e^{\frac{\mathrm{i}}{\hbar} p\cdot(w_1-u_1)} \nabla \phi(p)  \right) \cdot  \nabla V_N(q-q_2) \left[\gamma_{N,t}^{(2)}(u_1,w_2;w_1,w_2)  - \gamma^{(1)}_{N,t}(u_1; w_1) \gamma^{(1)}_{N,t}(w_2; w_2)\right] \bigg|\\
	& =: \text{I}_{2} + \text{J}_{2}.
	\end{align*}

Without the loss of generality, we let $\Phi(q,p) = \varphi(q)\phi(p)$. First, considering the term $\text{J}_{2}$,
\begin{align*}
	\text{J}_{2} & =  (2 \pi )^3\bigg| 	\iint \mathrm{d}q \mathrm{d}q_2\ \varphi(q) \iiint \mathrm{d}w_1\mathrm{d}u_1\mathrm{d}w_2\   f\left(\frac{w_1-q}{\sqrt{\hbar}}\right) f\left(\frac{u_1-q}{\sqrt{\hbar}}\right)  \bigg|f\left(\frac{w_2-q_2}{\sqrt{\hbar}}\right)\bigg|^2 \\
	&
	\qquad \left(\int \mathrm{d}p\ \rchi_{(w_1 - u_1)\in \Omega_\hbar^c} \nabla\phi(p) e^{\frac{\mathrm{i}}{\hbar} p\cdot(w_1-u_1)}\right) \cdot \nabla V_N(q-q_2) \left(\gamma_{N,t}^{(2)}(u_1,w_2;w_1,w_2)  - \gamma^{(1)}_{N,t}(u_1; w_1) \gamma^{(1)}_{N,t}(w_2; w_2)\right)\bigg|.
	\end{align*}
By the change of variable $\sqrt{\hbar}\tilde{q}_2 = w_2 - q_2$, we obtain
\begin{align*}
	& = (2 \pi )^3  \bigg|	\int \mathrm{d}q\ \varphi(q) \iiint \mathrm{d}w_1\mathrm{d}u_1\mathrm{d}w_2\  f\left(\frac{w_1-q}{\sqrt{\hbar}}\right) f \left(\frac{u_1-q}{\sqrt{\hbar}}  \right) \left(\hbar^{\frac{3}{2}} \int  \mathrm{d}\tilde{q}_2\   |f(\tilde{q}_2)|^2 \right) \nabla V_N(q-w_2 + \sqrt{\hbar}\tilde{q}_2) \\
	&
	\qquad \left(\int \mathrm{d}p\ \rchi_{(w_1 - u_1)\in \Omega_\hbar^c} \nabla  \phi(p) e^{\frac{\mathrm{i}}{\hbar} p\cdot(w_1-u_1)}  \right)\left(\gamma_{N,t}^{(2)}(u_1,w_2;w_1,w_2)  - \gamma^{(1)}_{N,t}(u_1; w_1) \gamma^{(1)}_{N,t}(w_2; w_2)\right) \bigg|\\
	&
	\leq C \norm{\nabla V_N}_{L^\infty}  \hbar^{\frac{3}{2}} 	\int \mathrm{d}q\  |\varphi(q)| \iiint \mathrm{d}w_1\mathrm{d}u_1\mathrm{d}w_2\  \left|f\left(\frac{w_1-q}{\sqrt{\hbar}}\right) f\left(\frac{u_1-q}{\sqrt{\hbar}}\right) \right|\\
	&
	\qquad \left|\int \mathrm{d}p\ \rchi_{(w_1 - u_1)\in \Omega_\hbar^c} \nabla  \phi(p) e^{\frac{\mathrm{i}}{\hbar} p\cdot(w_1-u_1)}\right| \left| \left(\gamma_{N,t}^{(2)}(u_1,w_2;w_1,w_2) -  \gamma_{N,t}(u_1; w_1) \gamma_{N,t}(w_2; w_2)\right) \right|\\
	&
	\leq C \norm{\nabla V_N}_{L^\infty}  \hbar^{\frac{3}{2}}   \int \mathrm{d}q\ | \varphi(q)| \iint \mathrm{d}w_1\mathrm{d}u_1\  \bigg| f\left(\frac{w_1-q}{\sqrt{\hbar}}\right) f\left(\frac{u_1-q}{\sqrt{\hbar}}\right)\bigg|\rchi_{|w_1-u_1|\leq 2R_1 \sqrt{\hbar}}\\
	&
	\qquad  \int\mathrm{d}w_2 \left|\gamma_{N,t}^{(2)}(u_1,w_2;w_1,w_2)  - \gamma^{(1)}_{N,t}(u_1; w_1) \gamma^{(1)}_{N,t}(w_2; w_2)\right|   \left|\int \mathrm{d}p\ \rchi_{(w_1 - u_1)\in \Omega_\hbar^c} \nabla  \phi(p) e^{\frac{\mathrm{i}}{\hbar} p\cdot(w_1-u_1)}\right|.\\
	\end{align*}
Recall again from Lemma \ref{estimate_oscillation} that we have
\begin{equation*}
	\left|\int \mathrm{d}p\ \rchi_{(w_1 - u_1)\in \Omega_\hbar^c} e^{\frac{\mathrm{i}}{\hbar} p\cdot(w_1-u_1)} \nabla  \phi(p) \right| \leq \norm{\nabla \phi}_{W^{s,\infty}}  \hbar^{(1-\alpha_2)s},
	\end{equation*}
for $s$ to be chosen later. Then, we obtain
\begin{align*}
	\text{J}_{2} & \leq C \norm{\nabla \phi}_{W^{s,\infty}} \norm{\nabla V_N}_{L^\infty}  \hbar^{\frac{3}{2}+ (1-\alpha_2)s}  \int \mathrm{d}q\ |\varphi(q)|\iint \mathrm{d}w_1\mathrm{d}u_1\  \bigg|  f\left(\frac{w_1-q}{\sqrt{\hbar}}\right) f\left(\frac{u_1-q}{\sqrt{\hbar}}\right)\bigg| \\
	&
	\qquad  \Tr^{(1)} \left|\gamma_{N,t}^{(2)}  - \gamma^{(1)}_{N,t} \otimes \gamma^{(1)}_{N,t} \right| (u_1;w_1)  \rchi_{|w_1-u_1|\leq 2R_1 \sqrt{\hbar}},
	\end{align*}
The H\"older inequality yields
\begin{align*}
	\text{J}_{2} & \leq C  \norm{\nabla \phi}_{W^{s,\infty}} \norm{\nabla V_N}_{L^\infty} \hbar^{\frac{3}{2}+ (1-\alpha_2)s} \int \mathrm{d}q\ |\varphi(q)|\left( \iint \mathrm{d}w_1\mathrm{d}u_1\ \rchi_{|w_1-u_1|\leq 2R_1 \sqrt{\hbar}} \left|f\left(\frac{w_1-q}{\sqrt{\hbar}}\right) f\left(\frac{u_1-q}{\sqrt{\hbar}}\right)\right|^2 \right)^\frac{1}{2}\\
	&
	\qquad \left(\iint \mathrm{d}w_1\mathrm{d}u_1\ \left[ \Tr^{(1)} \left|\gamma_{N,t}^{(2)}  -  \gamma^{(1)}_{N,t} \otimes \gamma^{(1)}_{N,t} \right|(u_1;w_1) \right]^2 \right)^\frac{1}{2}\\
	&
	=  C \norm{\varphi}_{L^\infty} \norm{\nabla \phi}_{W^{s,\infty}} \norm{\nabla V_N}_{L^\infty} \hbar^{\frac{3}{2}+ (1-\alpha_2)s} \left(\hbar^3 \iint \mathrm{d}\tilde{w}_1\mathrm{d}\tilde{u}_1\ \rchi_{|\tilde{w}_1-\tilde{u}_1|\leq 2R_1 } |f\left(\tilde{w}\right) f\left(\tilde{u}\right)|^2 \right)^\frac{1}{2}\\
	&
	\qquad \int \mathrm{d}q\ \left(\iint \mathrm{d}w_1\mathrm{d}u_1\ \left[ \Tr^{(1)} \left|\gamma_{N,t}^{(2)}  -  \gamma^{(1)}_{N,t} \otimes \gamma^{(1)}_{N,t} \right|(u_1;w_1) \right]^2  \rchi_{|w_1-q|\leq R_1 \sqrt{\hbar}} \right)^\frac{1}{2}\\
	&
	\leq C   \norm{\varphi}_{L^\infty} \norm{\nabla \phi}_{W^{s,\infty}} \norm{\nabla V_N}_{L^\infty}  \hbar^{ 3 + (1-\alpha_2)s } \left(\iint \mathrm{d}\tilde{w}_1\mathrm{d}\tilde{u}_1\ |f\left(\tilde{w}\right) f\left(\tilde{u}\right)|^2\right)^\frac{1}{2} \\
	&
	\qquad \hbar^{\frac{3}{2}} \int \mathrm{d}\tilde{q}_1 \rchi_{|\tilde{q}_1|\leq R_1 }   \left(\iint \mathrm{d}w_1\mathrm{d}u_1\ \left[ \Tr^{(1)} \left|\gamma_{N,t}^{(2)}  - \gamma^{(1)}_{N,t} \otimes \gamma^{(1)}_{N,t} \right|(u_1;w_1) \right]^2  \right)^\frac{1}{2}\\
	&
	\leq  \widetilde{C} \norm{\nabla V_N}_{L^\infty} \hbar^{ 3 + (1-\alpha_2)s+\frac{3}{2}} \left(\iint \mathrm{d}w_1\mathrm{d}u_1\ \left[ \Tr^{(1)} \left|\gamma_{N,t}^{(2)}  -  \gamma^{(1)}_{N,t} \otimes \gamma^{(1)}_{N,t} \right|(u_1;w_1) \right]^2  \right)^\frac{1}{2},
	\end{align*}
where we denote
\[
	\Tr^{(1)} \left|\gamma_{N,t}^{(2)}  - \gamma^{(1)}_{N,t} \otimes \gamma^{(1)}_{N,t} \right|(u_1;w_1) = \int \dd{w_2} |\gamma_{N,t}^{(2)}(u_1,w_2;w_1,w_2) - \gamma^{(1)}_{N,t}(u_1; w_1) \gamma^{(1)}_{N,t}(w_2; w_2)|.
	\]
Thus, we have
\begin{equation*}
	\text{J}_{2}  \leq \widetilde{C} \norm{\nabla V_N}_{L^\infty} \hbar^{ 3 + (1-\alpha_2)s+\frac{3}{2}}  \left(\iint \mathrm{d}w_1\mathrm{d}u_1\ \left[ \Tr^{(1)} \left|\gamma_{N,t}^{(2)}  -  \gamma^{(1)}_{N,t} \otimes \gamma^{(1)}_{N,t} \right|(u_1;w_1) \right]^2  \right)^\frac{1}{2}.
	\end{equation*}

Now, we focus on $\text I_{2,1}$
\begin{align*}
	\text{I}_{2,1}& =  (2 \pi )^3 \bigg| \iint (\mathrm{d}q)^{\otimes 2} \varphi(q) \iiint \mathrm{d}w_1\mathrm{d}u_1\mathrm{d}w_2\   f\left(\frac{w_1-q}{\sqrt{\hbar}}\right) f\left(\frac{u_1-q}{\sqrt{\hbar}}\right) \left|f\left(\frac{w_2-q_2}{\sqrt{\hbar}}\right)\right|^2  \\
	&
	\qquad \left(\int \mathrm{d}p\ \rchi_{(w_1 - u_1)\in \Omega_\hbar}  e^{\frac{\mathrm{i}}{\hbar} p\cdot(w_1-u_1)}  \nabla \phi(p) \right) \cdot \nabla V_N(q-q_2)\\
	&
	\qquad  \left[\gamma_{N,t}^{(2)}(u_1,w_2;w_1,w_2)  - \gamma^{(1)}_{N,t}(u_1; w_1) \gamma^{(1)}_{N,t}(w_2; w_2)\right] \bigg|.
	\end{align*}
We observe that
\[
	\left|\int \mathrm{d}p\  e^{\frac{\mathrm{i}}{\hbar} p\cdot(w_1-u_1)} \nabla  \phi(p)\right|\leq \norm{\nabla \phi}_{L^1}.
	\]
Then, we obtain the following estimate:
\begin{align*}
	\text{I}_{2}& \leq C \norm{\nabla \phi}_{L^1}\norm{\nabla V_N}_{L^\infty} \int \mathrm{d}q\ |\varphi(q)|   \left( \iint \mathrm{d}w_1\mathrm{d}u_1\ \rchi_{|w_1-u_1|\leq \hbar^\alpha_2}  \rchi_{|w_1-u_1|\leq 2R_1 \sqrt{\hbar}} \left|f\left(\frac{w_1-q}{\sqrt{\hbar}}\right) f\left(\frac{u_1-q}{\sqrt{\hbar}}\right)\right|^2 \right)^\frac{1}{2}\\
	&
	\qquad \hbar^{\frac{3}{2}} \int \dd{\tilde{q}_2} |f(\tilde{q}_2)|^2 \left(\iint \mathrm{d}w_1\mathrm{d}u_1\ \left[ \Tr^{(1)} \left|\gamma_{N,t}^{(2)}  - \gamma^{(1)}_{N,t} \otimes \gamma^{(1)}_{N,t} \right|(u_1;w_1) \right]^2  \rchi_{|w_1-q|\leq R_1 \sqrt{\hbar}}  \right)^\frac{1}{2}\\
	&
	\leq C \norm{\varphi}_{L^\infty} \norm{\nabla \phi}_{L^1} \norm{\nabla V_N}_{L^\infty}  \hbar^{\frac{3}{2}} \left(\hbar^3  \iint \mathrm{d}\tilde{w}_1\mathrm{d}\tilde{u}_1\ \rchi_{|\tilde{w}_1-\tilde{u}_1|\leq \hbar^{\alpha_2 + \frac{1}{2}}}  \rchi_{|\tilde{w}_1-\tilde{u}_1|\leq 2R_1} \left|f(\tilde{w}_1) f(\tilde{u}_1) \right|^2 \right)^\frac{1}{2}\\
	&
	\qquad \int\dd{q} \left(\iint \mathrm{d}w_1\mathrm{d}u_1\ \left[ \Tr^{(1)} \left|\gamma_{N,t}^{(2)}  - \gamma^{(1)}_{N,t} \otimes \gamma^{(1)}_{N,t} \right| (u_1;w_1) \right]^2  \rchi_{|w_1-q|\leq R_1 \sqrt{\hbar}}   \right)^\frac{1}{2}\\
	&
	\leq \widetilde{C} \norm{\nabla V_N}_{L^\infty} \hbar^3 \left( \iint \mathrm{d}\tilde{w}_1\mathrm{d}\tilde{u}_1\ \rchi_{|\tilde{w}_1-\tilde{u}_1|\leq \hbar^{\alpha_2 + \frac{1}{2}}}  \left|f(\tilde{w}_1) f(\tilde{u}_1) \right|^2 \right)^\frac{1}{2}\\
	&
	\qquad \hbar^{ \frac{3}{2}}\int \dd{\tilde{q}_1} \rchi_{|\tilde{q}_1|\leq R_1} 
	\left(\iint \mathrm{d}w_1\mathrm{d}u_1\ \left[ \Tr^{(1)} \left|\gamma_{N,t}^{(2)}  - \gamma^{(1)}_{N,t} \otimes \gamma^{(1)}_{N,t}  \right| (u_1;w_1) \right]^2 \right)^\frac{1}{2}.
	\end{align*}
From \eqref{bound_rchi_f}, we have
\begin{equation*}
	\text{I}_{2} \leq  \widetilde{C}  \norm{\nabla V_N}_{L^\infty} \hbar^{3+ \frac{3}{2}(\alpha_2 + \frac{1}{2}) + \frac{3}{2}} \left(\iint \mathrm{d}w_1\mathrm{d}u_1\ \left[ \Tr^{(1)} \left|\gamma_{N,t}^{(2)}  - \gamma^{(1)}_{N,t} \otimes \gamma^{(1)}_{N,t}  \right| (u_1;w_1) \right]^2  \right)^\frac{1}{2}.
	\end{equation*}

To balance the order between $\text{I}_{2}$ and $\text{J}_{2}$, $s$ is chosen to be
\[
	s = \left\lceil \frac{3(2\alpha_2 +1)}{4(1-\alpha_2)} \right\rceil,
	\]
for $\alpha_2 \in \left[0,1\right)$. Therefore, we have
\begin{equation*}
	\begin{aligned}
	\bigg|&\iint \mathrm{d}q \mathrm{d}p\ \ \varphi(q) \phi(p) \nabla_{p} \cdot \mathcal{R}_{2} \bigg| \leq \text{I}_{2} +\text{J}_{2} \\
	&
	\leq \widetilde{C} \norm{\nabla V_N}_{L^\infty}  \hbar^{3+ \frac{3}{2}(\alpha_2 + \frac{1}{2}) + \frac{3}{2}} \left(\iint \mathrm{d}w_1\mathrm{d}u_1\ \left[ \Tr^{(1)} \left|\gamma_{N,t}^{(2)}  - \gamma^{(1)}_{N,t} \otimes \gamma^{(1)}_{N,t}  \right|(u_1;w_1) \right]^2   \right)^\frac{1}{2}\\
	&
	\leq  \widetilde{C} \beta_N^{-2} \hbar^{3+ \frac{3}{2}(\alpha_2 + \frac{1}{2}) + \frac{3}{2}} N^2.
	\end{aligned}
	\end{equation*}
Setting $\beta_{N} = \hbar^\delta$ for $0 < \delta < \frac{1}{8}(6\alpha_1 - 5)$, we have the desired inequality.

\end{proof}

\vspace*{\fill}

\textit{Acknowledgements:}
 We are grateful to the anonymous referees for carefully reading our manuscript and providing helpful comments. Furthermore, we acknowledge support from the Deutsche Forschungsgemeinschaft through grant CH 955/4-1. Jinyeop Lee was partially supported by Samsung Science and Technology Foundation (SSTF-BA1401-51) and by National Research Foundation of Korea (NRF) grants funded by the Korean government (MSIT) (NRF-2019R1A5A1028324 and NRF-2020R1F1A1A01070580). Matthew Liew was financially supported by Landesgraduiertenf\"orderung of Baden-W\"urttemberg.

\clearpage
\bibliographystyle{abbrv}
\bibliography{husimi_vlasov}

\begin{thebibliography}{10}

\bibitem{ambrosio2014lagrangian}
L.~Ambrosio, M.~Colombo, and A.~Figalli.
\newblock On the lagrangian structure of transport equations: the
  {V}lasov--{P}oisson system.
\newblock {\em arXiv preprint arXiv:1412.3608}, 2014.

\bibitem{amour2013classical}
L.~Amour, M.~Khodja, and J.~Nourrigat.
\newblock The classical limit of the {H}eisenberg and time-dependent
  {H}artree--{F}ock equations: the {W}ick symbol of the solution.
\newblock {\em Mathematical Research Letters}, 20(1):119--139, 2013.

\bibitem{amour2013}
L.~Amour, M.~Khodja, and J.~Nourrigat.
\newblock The semiclassical limit of the time dependent {H}artree–-{F}ock
  equation: The {W}eyl symbol of the solution.
\newblock {\em Anal. PDE}, 6(7):1649--1674, 2013.

\bibitem{Arsenextquotesingleev1975}
A.~Arsen'ev.
\newblock Global existence of a weak solution of {Vl}asov's system of
  equations.
\newblock {\em {USSR} Computational Mathematics and Mathematical Physics},
  15(1):131--143, jan 1975.

\bibitem{Athanassoulis2011}
A.~Athanassoulis, T.~Paul, F.~Pezzotti, and M.~Pulvirenti.
\newblock Strong semiclassical approximation of wigner functions for the
  {H}artree dynamics.
\newblock {\em Atti della Accademia Nazionale dei Lincei, Classe di Scienze
  Fisiche, Matematiche e Naturali, Rendiconti Lincei Matematica E
  Applicazioni}, 22, 09 2011.

\bibitem{BACH20161}
V.~Bach, S.~Breteaux, S.~Petrat, P.~Pickl, and T.~Tzaneteas.
\newblock Kinetic energy estimates for the accuracy of the time-dependent
  {H}artree–{F}ock approximation with coulomb interaction.
\newblock {\em Journal de Mathématiques Pures et Appliquées}, 105(1):1--30,
  2016.

\bibitem{batt1977global}
J.~Batt.
\newblock Global symmetric solutions of the initial value problem of stellar
  dynamics.
\newblock {\em Journal of Differential Equations}, 25(3):342--364, 1977.

\bibitem{benediktermixed}
N.~Benedikter, V.~Jakšić, M.~Porta, C.~Saffirio, and B.~Schlein.
\newblock Mean-field evolution of fermionic mixed states.
\newblock {\em Communications on Pure and Applied Mathematics},
  69(12):2250--2303, 2016.

\bibitem{benedikter2016hartree}
N.~Benedikter, M.~Porta, C.~Saffirio, and B.~Schlein.
\newblock From the {H}artree dynamics to the {V}lasov equation.
\newblock {\em Archive for Rational Mechanics and Analysis}, 221(1):273--334,
  2016.

\bibitem{benedikter2014mean}
N.~Benedikter, M.~Porta, and B.~Schlein.
\newblock Mean--field evolution of fermionic systems.
\newblock {\em Communications in Mathematical Physics}, 331(3):1087--1131,
  2014.

\bibitem{benedikter2014rel}
N.~Benedikter, M.~Porta, and B.~Schlein.
\newblock Mean-field dynamics of fermions with relativistic dispersion.
\newblock {\em Journal of Mathematical Physics}, 55(2):021901, 2014.

\bibitem{bohun2016lagrangian}
A.~Bohun, F.~Bouchut, and G.~Crippa.
\newblock Lagrangian solutions to the {V}lasov--{P}oisson system with l1
  density.
\newblock {\em Journal of Differential Equations}, 260(4):3576--3597, 2016.

\bibitem{doi:10.1063/1.531326}
T.~Bröcker and R.~F. Werner.
\newblock Mixed states with positive {W}igner functions.
\newblock {\em Journal of Mathematical Physics}, 36(1):62--75, 1995.

\bibitem{CaseWignerpedestrians}
W.~Case.
\newblock {W}igner functions and {W}eyl transforms for pedestrians.
\newblock {\em American Journal of Physics}, 76, 10 2008.

\bibitem{Chen2021}
L.~Chen, J.~Lee, and M.~Liew.
\newblock Combined mean-field and semiclassical limits of large fermionic
  systems.
\newblock {\em Journal of Statistical Physics}, 182(2), jan 2021.

\bibitem{Chen2018}
L.~Chen, J.~O. Lee, and J.~Lee.
\newblock Rate of convergence toward {H}artree dynamics with singular
  interaction potential.
\newblock {\em Journal of Mathematical Physics}, 59(3):031902, 2018.

\bibitem{Chen2011}
L.~Chen, J.~O. Lee, and B.~Schlein.
\newblock Rate of convergence towards {H}artree dynamics.
\newblock {\em Journal of Statistical Physics}, 144(4):872, Aug 2011.

\bibitem{chong2021schrodinger}
J.~Chong, L.~Lafleche, and C.~Saffirio.
\newblock From {S}chr\"odinger to {H}artree--{F}ock and {V}lasov equations with
  singular potentials.
\newblock {\em arXiv preprint arXiv:2103.10946}, 2021.

\bibitem{Combescure2012}
M.~Combescure and D.~Robert.
\newblock {\em Coherent States and Applications in Mathematical Physics}.
\newblock Springer Netherlands, 2012.

\bibitem{Derezinski2009}
J.~Derezinski and C.~Gerard.
\newblock {\em Mathematics of Quantization and Quantum Fields}.
\newblock Cambridge University Press, 2009.

\bibitem{Dietler2018}
E.~Dietler, S.~Rademacher, and B.~Schlein.
\newblock From {H}artree dynamics to the relativistic {V}lasov equation.
\newblock {\em Journal of Statistical Physics}, 172(2):398--433, Jul 2018.

\bibitem{Dobrushin1979}
R.~L. Dobrushin.
\newblock {V}lasov equations.
\newblock {\em Functional Analysis and Its Applications}, 13(2):115--123, Apr
  1979.

\bibitem{ELGART20041241}
A.~Elgart, L.~Erdős, B.~Schlein, and H.-T. Yau.
\newblock Nonlinear {H}artree equation as the mean field limit of weakly
  coupled fermions.
\newblock {\em Journal de Mathématiques Pures et Appliquées}, 83(10):1241 --
  1273, 2004.

\bibitem{erdos2001derivation}
L.~Erdos and H.-T. Yau.
\newblock Derivation of the nonlinear schr\"odinger equation with {C}oulomb
  potential.
\newblock Technical report, 2001.

\bibitem{Fefferman1986}
C.~Fefferman and R.~de~la Llave.
\newblock Relativistic stability of matter-i.
\newblock {\em Revista Matematica Iberoamericana}, 2(2):119--213, 1986.

\bibitem{Fournais2018}
S.~Fournais, M.~Lewin, and J.~P. Solovej.
\newblock The semi-classical limit of large fermionic systems.
\newblock {\em Calculus of Variations and Partial Differential Equations},
  57(4):105, Jun 2018.

\bibitem{Frohlich2011}
J.~Fr{\"o}hlich and A.~Knowles.
\newblock A microscopic derivation of the time-dependent {H}artree-{F}ock
  equation with coulomb two-body interaction.
\newblock {\em Journal of Statistical Physics}, 145(1):23, Sep 2011.

\bibitem{Gasser1998}
I.~Gasser, R.~Illner, P.~A. Markowich, and C.~Schmeiser.
\newblock Semiclassical, $t\rightarrow \infty $ asymptotics and dispersive
  effects for {H}artree--{F}ock systems.
\newblock {\em ESAIM: Mathematical Modelling and Numerical Analysis -
  Mod\'elisation Math\'ematique et Analyse Num\'erique}, 32(6):699--713, 1998.

\bibitem{gilbarg2015elliptic}
D.~Gilbarg and N.~S. Trudinger.
\newblock {\em Elliptic partial differential equations of second order}, volume
  224.
\newblock springer, 2015.

\bibitem{Golse2016}
F.~Golse, C.~Mouhot, and T.~Paul.
\newblock On the mean field and classical limits of quantum mechanics.
\newblock {\em Communications in Mathematical Physics}, 343(1):165--205, Apr
  2016.

\bibitem{Golse2017}
F.~Golse and T.~Paul.
\newblock The {S}chr{\"o}dinger equation in the mean-field and semiclassical
  regime.
\newblock {\em Archive for Rational Mechanics and Analysis}, 223(1):57--94, Jan
  2017.

\bibitem{Golse2021}
F.~Golse and T.~Paul.
\newblock Mean-field and classical limit for the n -body quantum dynamics with
  {C}oulomb interaction.
\newblock {\em Communications on Pure and Applied Mathematics}, mar 2021.

\bibitem{golse:hal-01334365}
F.~Golse, T.~Paul, and M.~Pulvirenti.
\newblock {On the derivation of the {H}artree equation in the mean field limit:
  Uniformity in the Planck constant}.
\newblock {\em {Journal of Functional Analysis}}, 275(7):1603--1649, 2018.

\bibitem{hainzl2002general}
C.~Hainzl and R.~Seiringer.
\newblock General decomposition of radial functions on $r_n$ and applications
  to $n$-body quantum systems.
\newblock {\em Letters in Mathematical Physics}, 61(1):75--84, 2002.

\bibitem{jabin2011particles}
M.~Hauray and P.-E. Jabin.
\newblock Particle approximation of {V}lasov equations with singular forces:
  propagation of chaos.
\newblock {\em Ann. Sci. \'{E}c. Norm. Sup\'{e}r. (4)}, 48(4):891--940, 2015.

\bibitem{horst1981classical}
E.~Horst and H.~Neunzert.
\newblock On the classical solutions of the initial value problem for the
  unmodified non-linear {V}lasov equation i general theory.
\newblock {\em Mathematical Methods in the Applied Sciences}, 3(1):229--248,
  1981.

\bibitem{Hudson1974}
R.~Hudson.
\newblock When is the {W}igner quasi-probability density non-negative?
\newblock {\em Reports on Mathematical Physics}, 6(2):249 -- 252, 1974.

\bibitem{PhysRevLettKatz}
I.~Katz, A.~Retzker, R.~Straub, and R.~Lifshitz.
\newblock Signatures for a classical to quantum transition of a driven
  nonlinear nanomechanical resonator.
\newblock {\em Phys. Rev. Lett.}, 99:040404, Jul 2007.

\bibitem{kenfack2004negativity}
A.~Kenfack and K.~{\.Z}yczkowski.
\newblock Negativity of the {W}igner function as an indicator of
  non-classicality.
\newblock {\em Journal of Optics B: Quantum and Semiclassical Optics},
  6(10):396, 2004.

\bibitem{Lafleche2019GlobalSL}
L.~Lafleche.
\newblock Global semiclassical limit from {H}artree to {V}lasov equation for
  concentrated initial data.
\newblock {\em arXiv preprint arXiv:1902.08520}, 2019.

\bibitem{Lafleche2019PropagationOM}
L.~Lafleche.
\newblock Propagation of moments and semiclassical limit from {H}artree to
  {V}lasov equation.
\newblock {\em Journal of Statistical Physics}, 177(1):20--60, 2019.

\bibitem{lafleche2020strong}
L.~Lafl{\`e}che and C.~Saffirio.
\newblock Strong semiclassical limit from {H}artree and {H}artree--{F}ock to
  {V}lasov--{P}oisson equation.
\newblock {\em arXiv preprint arXiv:2003.02926}, 2020.

\bibitem{Lazarovici2017}
D.~Lazarovici and P.~Pickl.
\newblock A mean field limit for the {V}lasov--{P}oisson system.
\newblock {\em Archive for Rational Mechanics and Analysis}, 225(3):1201--1231,
  Sep 2017.

\bibitem{lieb1997thomas}
E.~H. Lieb.
\newblock Thomas-fermi and related theories of atoms and molecules.
\newblock {\em The Stability of Matter: From Atoms to Stars}, pages 259--297,
  1997.

\bibitem{Lieb2001}
E.~H. Lieb and M.~Loss.
\newblock {\em Analysis}.
\newblock American Mathematical Society, Providence, Rhode Island, 2001.

\bibitem{lieb1973thomas}
E.~H. Lieb and B.~Simon.
\newblock Thomas-fermi theory revisited.
\newblock {\em Physical Review Letters}, 31(11):681, 1973.

\bibitem{Lions1993}
P.-L. Lions and T.~Paul.
\newblock Sur les mesures de {W}igner.
\newblock {\em Revista Matemática Iberoamericana}, 9(3):553--618, 1993.

\bibitem{lions1991propagation}
P.-L. Lions and B.~Perthame.
\newblock Propagation of moments and regularity for the 3-dimensional
  {V}lasov--{P}oisson system.
\newblock {\em Inventiones mathematicae}, 105(1):415--430, 1991.

\bibitem{Loeper2006}
G.~Loeper.
\newblock Uniqueness of the solution to the {V}lasov{\textendash}{P}oisson
  system with bounded density.
\newblock {\em Journal de Math{\'{e}}matiques Pures et Appliqu{\'{e}}es},
  86(1):68--79, jul 2006.

\bibitem{PhysRevA.79.062302}
A.~Mandilara, E.~Karpov, and N.~J. Cerf.
\newblock Extending hudson's theorem to mixed quantum states.
\newblock {\em Phys. Rev. A}, 79:062302, Jun 2009.

\bibitem{Markowich1993}
P.~A. Markowich and N.~J. Mauser.
\newblock The classical limit of a self-consistent quantum-{V}lasov equation in
  3d.
\newblock {\em Mathematical Models and Methods in Applied Sciences},
  3:109--124, 1993.

\bibitem{Narnhofer1981}
H.~Narnhofer and G.~L. Sewell.
\newblock {V}lasov hydrodynamics of a quantum mechanical model.
\newblock {\em Communications in Mathematical Physics}, 79(1):9--24, Mar 1981.

\bibitem{petrat2014derivation}
S.~Petrat.
\newblock {\em Derivation of Mean-field Dynamics for Fermions}.
\newblock PhD thesis, 2014.

\bibitem{Petrat2017}
S.~Petrat.
\newblock {H}artree corrections in a mean-field limit for fermions with
  {C}oulomb interaction.
\newblock {\em Journal of Physics A: Mathematical and Theoretical},
  50(24):244004, may 2017.

\bibitem{Petrat2016ANM}
S.~Petrat and P.~Pickl.
\newblock A new method and a new scaling for deriving fermionic mean-field
  dynamics.
\newblock {\em Mathematical Physics, Analysis and Geometry}, 19:1--51, 2016.

\bibitem{pfaffelmoser1992global}
K.~Pfaffelmoser.
\newblock Global classical solutions of the {V}lasov--{P}oisson system in three
  dimensions for general initial data.
\newblock {\em Journal of Differential Equations}, 95(2):281--303, 1992.

\bibitem{Porta2017}
M.~Porta, S.~Rademacher, C.~Saffirio, and B.~Schlein.
\newblock Mean field evolution of fermions with coulomb interaction.
\newblock {\em Journal of Statistical Physics}, 166(6):1345--1364, 2017.

\bibitem{rodnianski2009quantum}
I.~Rodnianski and B.~Schlein.
\newblock Quantum fluctuations and rate of convergence towards mean field
  dynamics.
\newblock {\em Communications in Mathematical Physics}, 291(1):31--61, 2009.

\bibitem{saffirio2017mean}
C.~Saffirio.
\newblock Mean-field evolution of fermions with singular interaction.
\newblock In {\em Workshop on Macroscopic Limits of Quantum Systems}, pages
  81--99. Springer, 2017.

\bibitem{saffirio2019hartree}
C.~Saffirio.
\newblock From the {H}artree equation to the {V}lasov--{P}oisson system: Strong
  convergence for a class of mixed states.
\newblock {\em SIAM Journal on Mathematical Analysis}, 52(6):5533--5553, 2020.

\bibitem{Saffirio2019}
C.~Saffirio.
\newblock Semiclassical limit to the {V}lasov equation with inverse power law
  potentials.
\newblock {\em Communications in Mathematical Physics}, 373(2):571--619, 2020.

\bibitem{schaeffer1987global}
J.~Schaeffer.
\newblock Global existence for the {P}oisson-{V}lasov system with nearly
  symmetric data.
\newblock {\em Journal of differential equations}, 69(1):111--148, 1987.

\bibitem{doi:10.1063/1.525607}
F.~Soto and P.~Claverie.
\newblock When is the wigner function of multidimensional systems nonnegative?
\newblock {\em Journal of Mathematical Physics}, 24(1):97--100, 1983.

\bibitem{Spohn1981}
H.~Spohn.
\newblock On the {V}lasov hierarchy.
\newblock {\em Mathematical Methods in the Applied Sciences}, 3(1):445--455,
  1981.

\bibitem{Villani2003}
C.~Villani.
\newblock {\em Topics in Optimal Transportation}.
\newblock American Mathematical Society, 2003.

\bibitem{Zhang2008}
P.~Zhang.
\newblock {\em {W}igner Measure and Semiclassical Limits of Nonlinear
  {S}chrödinger Equations}.
\newblock American Mathematical Society, 2008.

\end{thebibliography}

\end{document}